\pdfoutput=1

\documentclass[11pt, draftcls, onecolumn]{IEEEtran}

\usepackage{setspace}
\usepackage{tikz}
\usetikzlibrary{shapes, arrows, positioning, calc}
\usepackage{graphicx,xcolor}
\usepackage[noadjust]{cite}
\usepackage{enumitem}

\usepackage{amsmath,amssymb,amsthm,amsfonts}
\usepackage{bbm}
\usepackage{bm}
\usepackage{mathrsfs}

\newtheorem{theorem}{Theorem}
\newtheorem{remark}{Remark}
\newtheorem{proposition}{Proposition}
\newtheorem{lemma}{Lemma}
\newtheorem{corollary}{Corollary}

\theoremstyle{definition}
\newtheorem{definition}{Definition}

\newcommand{\vf}{\mathbf{f}}
\newcommand{\f}{\mathrm{f}}
\newcommand{\vfx}{\mathbf{f}_X}
\newcommand{\fx}{\mathrm{f}_X}
\newcommand{\vfy}{\mathbf{f}_Y}

\newcommand{\vfz}{\mathbf{f}_Z}
\newcommand{\fz}{\mathrm{f}_Z}
\newcommand{\vfyh}{\mathbf{f}_{\hat{Y}}}
\newcommand{\fyh}{\mathrm{f}_{\hat{Y}}}
\newcommand{\vfzh}{\mathbf{f}_{\hat{Z}}}
\newcommand{\fzh}{\mathrm{f}_{\hat{Z}}}

\allowdisplaybreaks

\newcommand{\mc}[1]{\mathcal{#1}}
\newcommand{\msc}[1]{\mathscr{#1}}
\newcommand{\op}[1]{\mathop{\operatorname{#1}}}

\newcommand{\rk}{R_{K}}
\newcommand{\trk}{R'_{K}}

\newcommand{\bW}{\mathbf{W}}
\newcommand{\bX}{\mathbf{X}}
\newcommand{\bY}{\mathbf{Y}}

\newcommand{\ux}{\underline{x}}
\newcommand{\uy}{\underline{y}}

\newcommand{\uX}{\underline{X}}
\newcommand{\uY}{\underline{Y}}
\newcommand{\uZ}{\underline{Z}}
\newcommand{\hX}{\hat{X}}
\newcommand{\hY}{\hat{Y}}
\newcommand{\hZ}{\hat{Z}}
\newcommand{\uhX}{\underline{\hat{X}}}
\newcommand{\uhY}{\underline{\hat{Y}}}

\newcommand{\epm}{\epsilon_\text{m}}
\newcommand{\epw}{\epsilon_\text{w}}
\newcommand{\tR}{\tilde{R}}
\newcommand{\tW}{\tilde{W}}
\newcommand{\tw}{\tilde{w}}
\newcommand{\code}{\mc{C}_{ml}}
\newcommand{\randcode}{\msc{C}_{ml}}
\newcommand{\E}{\mathbb{E}}

\newcommand{\stv}{\mathrm{S}_{\op{\scriptscriptstyle TV}}}
\newcommand{\skl}{\mathrm{S}_{\op{\scriptscriptstyle KL}}}
\newcommand{\dtv}{d_{\op{\scriptscriptstyle TV}}}

\graphicspath{{./Figures/}}

\title{The Secure Storage Capacity \\ of a DNA Wiretap Channel Model}

\author{Praneeth Kumar Vippathalla and Navin Kashyap \thanks{ Praneeth Kumar V.\ (praneethv@iisc.ac.in) and N.\ Kashyap (nkashyap@iisc.ac.in)  are with the Department of Electrical Communication Engineering, Indian Institute of Science, Bangalore 560012. This work was supported in part by a grant, IE/RERE-19-0514, from the Indian Institute of Science. It was presented in part at the 2022 IEEE International Symposium on Information Theory (ISIT 2022), Espoo, Finland.}}

\begin{document}
\IEEEoverridecommandlockouts
\maketitle
\begin{abstract}
  In this paper, we propose a strategy for making DNA-based data storage information-theoretically secure through the use of wiretap channel coding. This motivates us to extend the shuffling-sampling channel model of Shomorony and Heckel (2021) to include a wiretapper. Our main result is a characterization of the \emph{secure storage capacity} of our DNA wiretap channel model, which is the maximum rate at which data can be stored within a pool of DNA molecules so as to be reliably retrieved by an authorized party (Bob), while ensuring that an unauthorized party (Eve) gets almost no information from her observations. Furthermore, our proof of achievability shows that index-based wiretap channel coding schemes are optimal.

\end{abstract} 


\section{Introduction} \label{sec:introduction}
In DNA-based storage, raw data (e.g., text) is converted into sequences over an alphabet consisting of the building blocks of DNA, namely, the four nucleotide bases, adenine (A), guanine (G), thymine (T), and cytosine (C). The sequences over the DNA alphabet are then physically realized by artificially synthesizing DNA molecules (called oligonucleotides, or oligos, in short) corresponding to the string of A, G, T, C letters forming the sequences. These synthetic DNA molecules (oligos) can be mass-produced and replicated to make many thousands of copies and the resulting pool of oligos can be stored away in a controlled environment.  At the time of data retrieval, sequencing technology is used to determine the A-G-T-C sequence forming each oligo from the pool. Prior to sequencing, the pool of oligos is subjected to several cycles of Polymerase Chain Reaction (PCR) amplification. In each cycle of PCR, each molecule in the pool is replicated (``amplified'') by a factor of 1.6--1.8. The PCR amplification process requires knowledge of short initial segments (prefixes) and final segments (suffixes) of the oligos to be amplified. This knowledge is used to design the  required primers to initiate PCR amplification. After PCR, a small amount of material from the amplified pool is passed through a sequencing platform \cite{goodwin16} that randomly samples and sequences (by patching together ``reads'' of a relatively short length) the molecules from the pool. The raw data is then retrieved from these sequences. 

While DNA-based data storage technology provides a reliable solution for long-term data storage, there often can arise security issues. Suppose Alice wants to store using DNA-based data storage some sensitive information which is to be retrieved later by a trusted party, Bob. A solution to this problem is that Alice uses a private key $K$, which is shared with Bob, to one-time pad her information $W$ and then store it into a pool of synthetic DNA oligonucleotides. But a significant drawback of this scheme is that the size of the key $K$ should be comparable to the size of information $W$, which is not practically feasible.  Another simple solution, proposed by Clelland \emph{et al}. \cite{Cle99}, is to design the oligos that encode $W$ in such a way that the keystring $K$ can be converted to the specific primers needed for initiating PCR-based amplification of these oligos. Then, these oligos can be synthesized in small amounts (low copy numbers), and then hidden within an ocean of ``background''  or ``junk'' DNA. The background DNA could, for example, be fragments from the human genome.  A vial containing the composite pool of information-bearing and background DNA is stored in a DNA-storage repository.   

Since Bob knows $K$, given the vial containing the composite oligo pool, he can provide the primers needed to selectively amplify only the information-bearing oligos using PCR. After several cycles of PCR, the copy numbers of these oligos in the post-PCR sample become large enough that when the sample is fed into a sequencing platform, each of these oligos get a large number of reads. After filtering out the reads that come from the background DNA (this is easily possible if the background DNA is made up of human genome fragments), the remaining reads correspond to the information-bearing oligos. These reads can be assembled using standard sequence assembly algorithms, after which their information content, $W$, can be recovered by Bob.

Eve, on the other hand, does not know $K$, so is unable to selectively amplify the information-bearing oligos. Instead, her only option is to amplify the entire composite oligo pool, including background DNA, using whole-library PCR. (In whole-library PCR, known adapters are non-selectively ligated onto the oligos in the pool, and the complementary sequences of these adapters are used as primers to initiate PCR.)  However, this process does not discriminate between information-bearing oligos and background DNA, so it does not significantly change the proportion of information-bearing molecules to background DNA. Since background DNA still constitutes the overwhelming majority of molecules in the amplified pool, it is highly unlikely that the information-bearing oligos will get sufficiently many reads to get reliably sequenced. This protects the data $W$ from being reliably recovered by Eve. 

\begin{figure}[h]
\centering
\includegraphics[width=0.7\columnwidth]{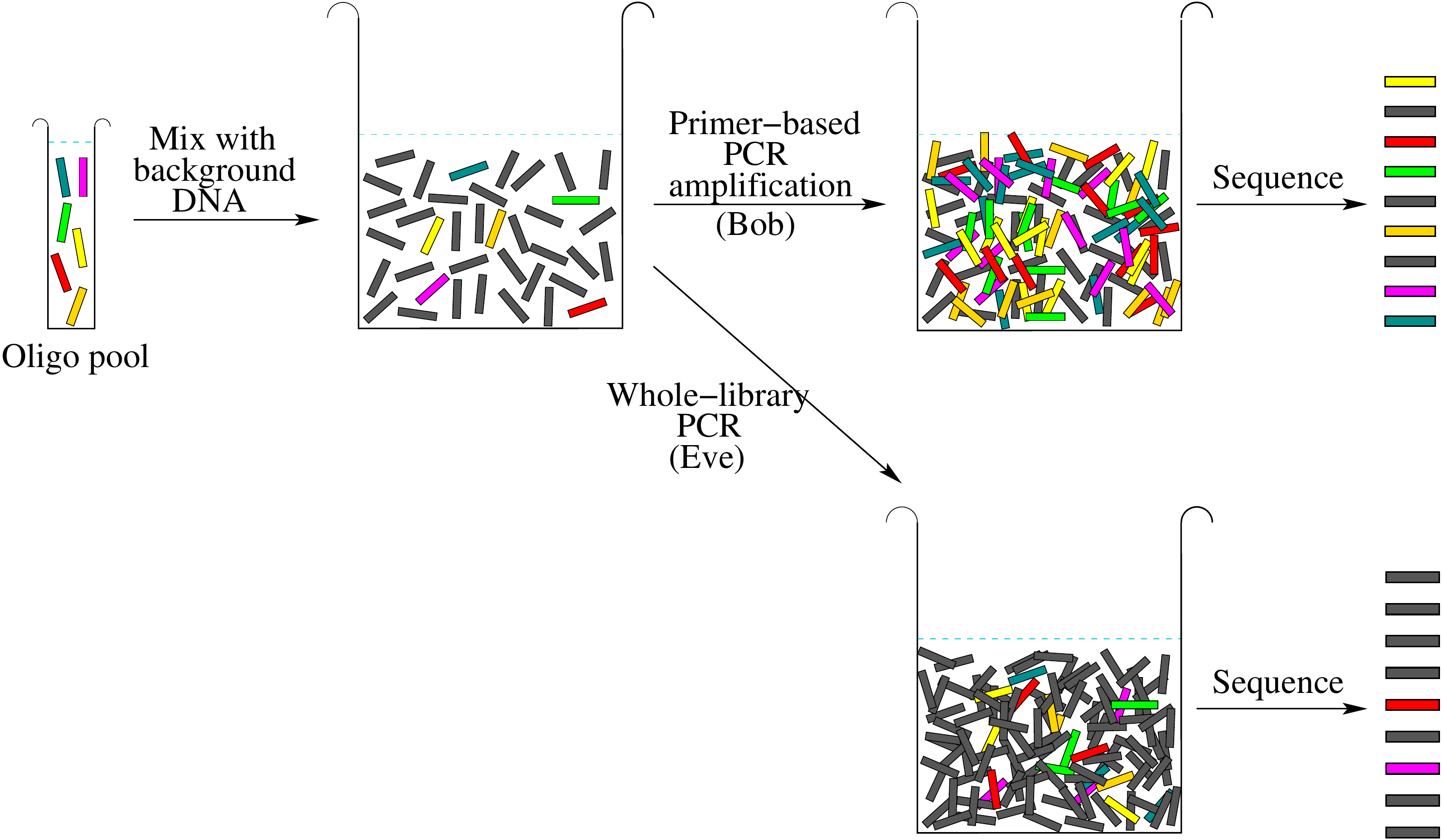}
\caption{The basic strategy of secure DNA-based data storage.}
\label{fig:dna_storage}
\end{figure}

One issue is that Eve's protocol will allow her to partially recover the data $W$, which may be undesirable. This is possible even after just one round of sequencing, but if her budget allows for multiple rounds of sequencing, then she can use the sequences recovered from the first round to determine some of the primers. She can then use these primers for selective amplification in the next round of sequencing to boost her chances of getting more information about $W$. The more the rounds of sequencing she is allowed, the better she can do.

\subsection{Our proposed strategy: Using coding to boost security}

What we propose is the use of coding to greatly (and cheaply) improve upon the basic scheme above. As observed above, the use of the shared key $K$ as primers gives Bob an advantage over Eve. In information-theoretic parlance, we say that the use of primers creates a ``channel'' from Alice to Bob that is much less noisy (much more reliable) than the ``channel'' from Alice to Eve. In the area of information-theoretic security, this situation is called a \emph{wiretap channel} \cite{Wyner75}, and the design of coding schemes for such a channel is a well-investigated research topic -- see e.g., \cite{csiszar78,bloch_barros_2011,Thangaraj07}. These coding schemes greatly boost any advantage that Bob has over Eve, however small it may initially be. To clarify, the aim of a wiretap channel coding scheme is to enable Alice to encode the data $W$ into a codeword $X$ (in our case, this will be realized as a pool of oligos) such that the difference in noise statistics between the Alice-Bob and Alice-Eve channels can be exploited, as envisioned below:
\begin{itemize}
\item[(a)] When $X$ passes through the channel from Alice to Bob, what Bob receives is a slightly noisy version of $X$, which we denote by $Y$, from which he is able to recover the data $W$ highly reliably;
\item[(b)] but when $X$ passes through the channel from Alice to Eve, Eve observes a significantly noisier version $Z$, from which she gets almost \emph{no information} about the data $W$.
\end{itemize}

The proportion, $\rho$, of information-bearing molecules in the composite pool is a knob that we can use to control the noise statistics of Eve's channel --- see Appendix~\ref{app:erasure_probs} for a brief discussion of how $\rho$ affects Eve's probability of oligo erasure. It is desirable to keep $\rho$ small ($10^{-3}$ or lower) so as to put Eve at a significant disadvantage. But it cannot be so small that the primers used to initiate PCR fail to find a sufficient quantity of information-bearing molecules to amplify, causing even Bob's protocol to fail.

The advantages of our approach over the basic approach are two-fold:
\begin{itemize}
\item[(i)] Our scheme should be able to tolerate higher ratios $\rho$ than the basic scheme of Clelland et al.\ \cite{Cle99}, as it is enough to choose a $\rho$ that gives a minor statistical advantage for Bob over Eve. Higher ratios $\rho$ will allow Bob's protocol to work much more reliably. 
\item[(ii)] Wiretap channel coding is supposed to ensure that Eve gets essentially no information about the data $W$, \emph{no matter what strategy she uses} to recover $W$ from her observations $Z$.
\end{itemize}
\subsection{Main contribution}
We build on the recent pioneering work of Shomorony and Heckel \cite{shomorony21} that models and determines the fundamental limits of DNA-based data storage (but without any considerations of security). In this work, data $W$ is encoded in the form of an oligo pool, which is viewed as a multiset $X$, each element of which is a sequence of length $L$ over the DNA alphabet $\Sigma = \{A,C,G,T\}$. From the multiset $X$, a multiset of $N$ sequences is randomly drawn according to some fixed probability distribution $P$. The $P$ distribution encapsulates the randomness inherent in the processes of oligo synthesis, PCR amplification and sequencing. The resulting multiset, $Y$, of sequences is observed by the decoder, which attempts to retrieve $W$ from $Y$. This channel is referred to as a \emph{(noise-free) shuffling-sampling channel} with distribution $P$. Shomorony and Heckel \cite{shomorony21} studied the \emph{storage capacity} of this model, defined to be the maximum rate at which data $W$ can be stored in an oligo pool $X$, and reliably retrieved from the observed multiset $Y$. Their study in fact extends to a \emph{noisy} model wherein the $N$ sequences sampled from $X$ may further be corrupted by insertion, deletion and substitution errors affecting the individual letters making up each sequence. 

\begin{figure}
\centering
\resizebox{0.75\columnwidth}{!}{\input{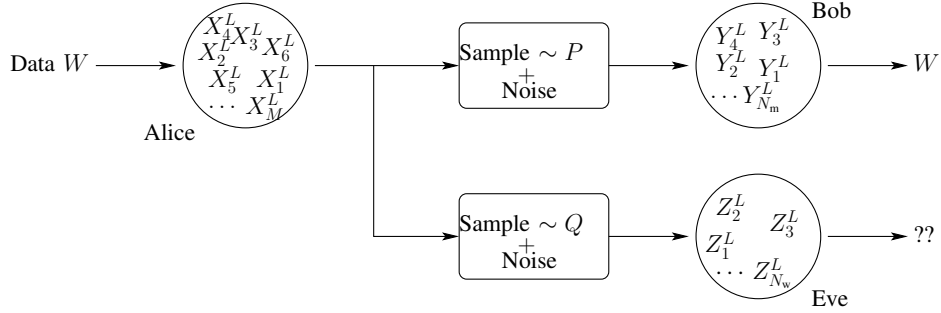}}
\caption{Our DNA storage wiretap channel model.}
\label{fig:dna_wtc_model}
\end{figure}

We extend the noise-free Shomorony-Heckel model above by introducing an additional noise-free shuffling-sampling channel with distribution $Q$ for the unauthorized party (Eve) 
--- see Fig.~\ref{fig:dna_wtc_model}. 
We study the \emph{secure storage capacity}, $C_s$, of the noise-free shuffling-sampling model, which we define to be the maximum rate at which data $W$ can be stored in an oligo pool $X$, and reliably retrieved from Bob's observation $Y$ and $K$, while ensuring that Eve gets (almost) no information about $W$ from her observations. We completely characterize $C_s$, giving an expression for it that depends only on $q_0$ and $p_0$. 
We also study the secure storage capacity of the noise-free shuffling-sampling model with a secret key $K$ of rate $\rk$ shared between Alice and Bob. We give a precise characterization of the secure storage capacity as a function of $\rk$. 

\subsection{Related work}

DNA-based data storage, and various situations that arise in communication networks, recently spurred interest in an independent study of the permutation (shuffling) channel \cite{anuran_permutation, kovacevic_2018, anshoo_2019}. In the work \cite{anuran_permutation}, Makur carried out an information-theoretic study of   noisy permutation channels and established some coding theorems. A noisy permutation channel acts on a sequence of symbols in the following way: first, the channel permutes the order of the  symbols, and then corrupts each symbol by passing it through a discrete memoryless channel. 
The main difference between this model and the one considered in \cite{shomorony21} is that the size of the alphabet of each symbol is fixed here, while it grows with the length of the sequence\footnote{In \cite{shomorony21}, we can view each oligo of length $L$ as a symbol taking values in the alphabet $\Sigma^L$.} in \cite{shomorony21}. Since the order of the sequences gets lost in the channel, we can encode messages, without sacrificing performance, into different multisets (empirical distributions or frequency vectors), and transmit representative sequences corresponding to each multiset across the channel. A study of codes over multisets, correcting insertions, deletions, and substitution errors, was done by Kova\v{c}evi\'c and Tan in \cite{kovacevic_2018}. Permutation channels also appear in the modeling of the bee-identification problem, which was studied by Tandon, Tan, and Varshney  in  \cite{anshoo_2019}.

In the seminal paper \cite{shomorony21}, Shomorony and Heckel proposed an information-theoretic model for DNA-based data storage and studied the fundamental limits on reliable data storage in the DNA medium. This model, which they called a \emph{noisy shuffling-sampling} channel, captures the randomness that arises in DNA-based data storage.  For a noise-free version of this channel, they exactly characterized the storage capacity, which is the maximum rate at which a message can be reliably stored in the oligo pool. Moreover, \cite{shomorony21}
considered an extension of the noise-free version, wherein a certain probabilistic model was assumed for substitution errors. Under some constraints on the parameters of this noisy model, Shomorony and Heckel gave an exact characterization of the storage capacity. Other works \cite{nir22, lenz19} studied the noisy shuffling-sampling channel under a different sampling model.  An experimental study of the probabilities of substitution, insertion and deletion errors at each phase (synthesis, amplification, and sequencing) can be found in the work of Heckel, Mikutis, and Grass \cite{heckel_errors}. 

A work that studies privacy in the context of DNA sequencing is \cite{mazooji_2020}. The idea of mixing 
a DNA sample with ``noise" DNA samples (known only to the individual preparing the mixed sample) is used in \cite{mazooji_2020} to hide the sequenced genetic information of an individual from a sequencing laboratory. The focus of \cite{mazooji_2020} is to find an optimal mixing strategy to maximize privacy. In our current work, we study information-theoretic security in the context of DNA-based data storage; to our knowledge, this has not been addressed in any prior published work.\footnote{A preliminary study of secure DNA-based data storage was reported in the Master's thesis of Johns U.K. \cite{johns_report}, supervised by the second author (N.K.).}

\subsection{Notation}
A useful fact about multisets is that they can be uniquely identified with frequency vectors. Given a multiset $A$ whose elements are from a finite set $\mathcal{A}=\{a_1,\ldots, a_{|\mathcal{A}|}\}$, we denote by $\vf_A=[\f_A(a_1), \ldots, \f_A(a_{|\mathcal{A}|})]$ the frequency vector corresponding to $A$. For $a_i \in  \mathcal{A}$, the component $\f_A(a_i)$ counts the number of occurrences of $a_i$ in the set $A$.  The components of a frequency vector will be typeset in regular font, for example $\fx$. The symbols $\mathbb{N}$ and $\mathbb{R}$ denote the set of natural numbers $\{1,2,3,\ldots \}$ and the set of real numbers, respectively. Let $h(x)$ denote the binary entropy function, i.e, $h(x)= -x \log x-(1-x) \log(1-x)$, for $x \in (0,1)$. The $n$-dimensional probability simplex, denoted by $\Delta^n$, is defined as 
\begin{align*}
    \Delta^n := \Big\{(x_0,x_1,\ldots,x_n)\in \mathbb{R}^{n+1}: \sum\limits_{i=0}^n x_i =1 \Big\}.
\end{align*}
For a positive integer $n$, we use the notation $[n]$ to denote the set $\{1,2,\ldots,n\}$. For real-valued sequences $(a_n)_{n\geq 1}$ and $(b_n)_{n\geq 1}$ with $b_n>0$ for all $n$ large enough, we write $a_n=O(b_n)$ if $\limsup_{n \to \infty} |a_n|/b_n<\infty$;  $a_n=o(b_n)$ if $\lim_{n \to \infty} a_n/b_n=0$; and $a_n \sim b_n$ if $\lim_{n\to \infty}a_n/b_n = 1$.
\subsection{Organization}
The rest of this paper is organized as follows. In Section~\ref{sec:problem},  we formally introduce a DNA storage wiretap channel model and define a notion of secure storage capacity for this channel model. In the same section, we present one of our main results on the characterization of the secure storage capacity of our model. Next, in Section~\ref{sec:problem_shared_key}, we study an extension of the DNA wiretap channel model, allowing Alice and Bob to use a shared key in the wiretap channel coding scheme. Section~\ref{sec:proof_main_result} contains proofs of our main results, with some of the auxiliary results relegated to Appendices~\ref{app:A}--\ref{app:block-erasure}. Appendix~\ref{app:block-erasure}, in particular, is devoted to the so-called block-erasure wiretap channel, which plays a crucial role in the achievability arguments for our theorems. 
In Appendix~\ref{app:erasure_probs}, we make some observations concerning oligo erasure probabilities when sequencing a mixed sample.

\section{A DNA Storage Wiretap Channel Model}\label{sec:problem}
Let $M$ denote the number of DNA molecules in the oligo pool, and $L$ denote the length of each molecule.
Though $\{A,C,G,T\}$ is the alphabet of DNA coding, we work with $\Sigma:=\{0,1\}$ for the sake of simplicity. However, in the case of  general alphabet $\Sigma$, the results will involve an extra $\log_2 |\Sigma|$ term.
Fix a constant $\beta := \lim_{M \to \infty} \frac{L}{\log M} > 1$; throughout this paper, all logarithms are to the base $2$. Let $W$ be a uniform random variable taking values in $\mathcal{W}\triangleq \{1, 2, \ldots, 2^{MLR}\}$ for some $R \geq 0$. Alice encodes (maps) the message $W$ into $M$ DNA molecules, each of length  $L$, which is denoted by a  multiset $X^{ML}=\{X_1^L, \ldots,X_M^L\}$. The stored $X^{ML}$ is amplified and sequenced to recover the message $W$. A fundamental model that captures the cumulative effect of these processes without errors is the \emph{noise-free shuffling-sampling channel} with some distribution $(\pi_0,\pi_1,\ldots)$ \cite{shomorony21}. This channel randomly permutes (shuffles) the order of the $L$-length molecules, and independently outputs each molecule $n \geq 0$ times with probability $\pi_n$. The output of this channel is also a multiset.  We can use this model for the amplification and synthesis at both Bob and Eve's side. Hence Bob and Eve observe multisets $Y^{N_{\text{m}}L}=\{Y_1^L, \ldots,Y_{N_{\text{m}}}^L\}$ and $Z^{N_{\text{w}}L}=\{Z_1^L, \ldots,Z_{N_{\text{w}}}^L\}$, respectively, which are obtained by passing the input $X^{ML}$ through two independent noise-free shuffling-sampling channels with distributions $P=(p_0,p_1,\ldots)$ and  $Q=(q_0,q_1,\ldots)$, respectively. See Fig.~\ref{fig:dna_wtc}. The quantities\footnote{The subscripts `m' and `w' stand for ``main'' and ``wiretap'', respectively.} $N_{\text{m}}$ and $N_{\text{w}}$ are random variables taking non-negative integer values, and $Y_i^L, Z_j^L$ are  elements of $\Sigma^L$. The goal of Alice is to store a message that can only be recovered by Bob while keeping Eve in ignorance of it. 
Formally, let $\phi: \mathcal{W} \longrightarrow {\mathbb{N}}^{\Sigma^L}$ be an encoding function (possibly a stochastic function) of Alice, and $\psi: {\mathbb{N}}^{\Sigma^L} \longrightarrow \mathcal{W} \cup \{\textbf{e}\}$ be a decoding function of Bob, where \textbf{e} denotes an error, and ${\mathbb{N}}^{\Sigma^L}$ denotes the set of all  multisets with finite cardinality over $\Sigma^L$.  We say that a secure message rate $R$ is achievable if there exists a sequence of pairs of encoding and decoding functions $\{(\phi,\psi)\}_{M=1}^{\infty}$ that satisfy Bob's recoverability condition,
 \begin{align} \label{eq:bob_recoverability}
     \mathbb{P}\{\psi(Y^{N_{\text{m}}L}) \neq W\} \to 0,
 \end{align}
 and the (strong) secrecy condition, \begin{align}\label{eq:eve_secrecy}
    I(W ; Z^{N_{\text{w}}L}) \to 0
 \end{align}
 as $M \to \infty$.
 The \emph{secure storage capacity}, $ C_s $, is defined by 
 \begin{align}
     C_s \triangleq \sup\{R: R \text{ is achievable}\}.
 \end{align}

\begin{figure}[h]
\centering
\resizebox{0.85\width}{!}{

    


\tikzstyle{block}=[rectangle, draw, thick, minimum width=2em, minimum height=2em]

\begin{tikzpicture}[node distance=3.5cm,auto,>=latex']
    \node (x) {$X^{ML}$};
    \node (dummy) [right of=x, node distance=3.5cm] {};
    \node [block, align=center] (main) [above of = dummy, node distance=1cm] {Noise-free \\ shuffling-sampling channel\\ $P=(p_0,p_1,p_2,\ldots)$};
    \node [block, align=center] (wiretap) [below of = dummy, node distance=1cm] {Noise-free \\ shuffling-sampling channel\\ $Q=(q_0,q_1,q_2,\ldots)$};
    \node (y) [right of=main] {$Y^{N_{\text{m}}L}$};
    \node (z) [right of=wiretap] {$Z^{N_{\text{w}}L}$};
        
    \node (w) [left of=x, node distance=1.5cm] {$W$};
       
    \node (w_hat) [right of=y, node distance=1.5cm] {$\widehat{W}$};
        
     \node (alice) [above of=w, node distance=1.2cm] {Alice};
     \node (bob) [above of= y, node distance=1cm] {Bob};
     \node (eve) [below right of= z, node distance=0.8cm] {Eve};
    
    \draw[->, thick] (x.east) --  (main.west);
    \draw[->, thick] (x.east) --  (wiretap.west);
    \draw[->, thick] (main) --  (y);
    \draw[->, thick] (wiretap) --  (z);
     \draw[->, thick] (w) -- (x.west);
    \draw[->, thick] (y.east) -- (w_hat);
   
\end{tikzpicture}}
\caption{A DNA storage wiretap channel model.}
\label{fig:dna_wtc}
 \end{figure}
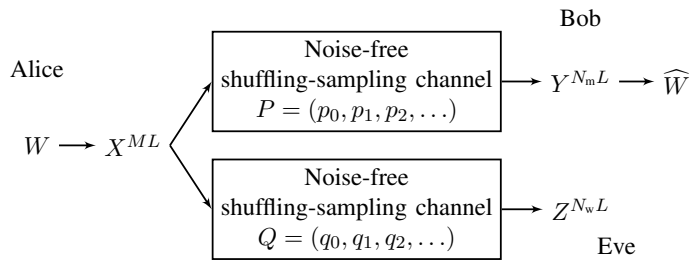

In a slight abuse of notation, we use $\vfx$, $\vfy$
and $\vfz$ to denote the frequency (random) vectors corresponding to the multisets $X^{ML},Y^{N_{\text{m}}L}$ and $Z^{N_{\text{w}}L}$, respectively, over $\Sigma^L$. 
As frequency vectors and multisets are interchangeable\footnote{Since a frequency vector gives an equivalent representation of a multiset, we have that for a random multiset $A$, $H(\vf_A|A)=0=H(A|\vf_A)$.}, the Markov chain $W-X^{ML}-(Y^{N_{\text{m}}L}, Z^{N_{\text{w}}L})$
is equivalent to
\begin{align}\label{mc:original_model}
    W-\vfx-(\vfy,\vfz)
\end{align}
for any joint distribution that is induced by an encoding function.

In the $q_0=1$ case, where none of the molecules are sampled by Eve's channel, the secure storage capacity is nothing but the storage capacity of Bob's channel \cite{shomorony21}.

\begin{theorem}[\cite{shomorony21}, Theorem~1]
The storage capacity of a noise-free shuffling-sampling channel with the distribution $(\pi_0,\pi_1,\ldots)$ is equal to $\left(1-\frac{1}{\beta}\right)(1-\pi_0)$. 
\end{theorem}

One of our main results is the following expression for the secure storage capacity of our model.
\begin{theorem}\label{thm:dna_wtc_capacity}
For a DNA storage wiretap channel with $q_0 \geq p_0$,
\begin{align}
    C_s = \left(1-\frac{1}{\beta}\right)(q_0-p_0).
\end{align}
\end{theorem}

\begin{proof}
    See Section~\ref{sec:proof_main_result}.
\end{proof}

\section{A DNA Storage Wiretap Channel Model with a Shared Key}\label{sec:problem_shared_key}
In addition to using a shared key as primers, if we allow Alice and Bob to use another shared key for the purpose of wiretap channel coding, then the resulting model is called a \emph{DNA storage wiretap channel model with a shared key.} The main difference between these two shared keys is that the former key creates the DNA storage wiretap channel model on which Alice and Bob use coding to boost security, whereas the latter is used solely for wiretap channel coding. Bearing this distinction in mind, we study in this section the effect of the latter shared key on the secure storage capacity of the DNA storage wiretap channel considered in Section~\ref{sec:problem}. 

For some $\rk \geq 0$, let $K \in \mathcal{K}_M\triangleq \{1, 2, \ldots, 2^{ML\rk}\}$  be a uniform random variable, called a \emph{shared key}. Only Alice and Bob have access to $K$.  Alice uses $K$ to encode a message $W$, which is independent of $K$, into a multiset $X^{ML}$.  Bob and Eve observe multisets $Y^{N_{\text{m}}L}$ and $Z^{N_{\text{w}}L}$, respectively, which are obtained by passing the input $X^{ML}$ through two independent noise-free shuffling-sampling channels with distributions $P=(p_0,p_1,\ldots)$ and  $Q=(q_0,q_1,\ldots)$, respectively. Using $K$, Bob tries to decode the message from $Y^{ML}$. Let $\phi: \mathcal{W} \times \mc{K}_M \longrightarrow {\mathbb{N}}^{\Sigma^L}$ be an encoding function (possibly a stochastic function) of Alice, and $\psi: {\mathbb{N}}^{\Sigma^L} \times \mc{K}_M \longrightarrow \mathcal{W} \cup \{\textbf{e}\}$ be a decoding function of Bob, where \textbf{e} denotes an error, and ${\mathbb{N}}^{\Sigma^L}$ denotes the set of all  multisets with finite cardinality over $\Sigma^L$.  We say that a secure message rate $R$ is achievable with a shared key of rate $\rk$ if there exists a sequence of pairs of encoding and decoding functions $\{(\phi,\psi)\}_{M=1}^{\infty}$ that uses a sequence of shared keys $\{K\}_{M=1}^{\infty}$, and satisfy Bob's recoverability condition,
 \begin{align} \label{eq:bob_recoverability}
     \mathbb{P}\{\psi(Y^{N_{\text{m}}L}, K) \neq W\} \to 0,
 \end{align}
 and the (strong) secrecy condition, \begin{align}\label{eq:eve_secrecy}
    I(W ; Z^{N_{\text{w}}L}) \to 0
 \end{align}
 as $M \to \infty$.
 The secure storage capacity of a DNA storage wiretap channel with a shared key of rate $\rk$ is defined by 
 \begin{align}
     C_s(\rk) \triangleq \sup\{R: R \text{ is achievable with a shared key of rate $\rk$}\}.
 \end{align}


The following theorem characterizes $C_s(\rk)$ of a DNA storage wiretap channel.
\begin{theorem}\label{thm:dna_wtc_key_capacity}
For a DNA storage wiretap channel with a shared key of rate $\rk$,
\begin{align}
    C_s(\rk) = \min \left\lbrace\left(1-\frac{1}{\beta}\right)[q_0-p_0]^{+}+\rk,\  C_M\right\rbrace,
\end{align}
where $[x]^{+}=\max\left\lbrace x, 0\right\rbrace$ and $C_M :=\left(1-\frac{1}{\beta}\right)(1-p_0)$ is the storage capacity of the main (Bob's) channel.
\end{theorem}
\begin{proof}
    See Section~\ref{sec:proof_main_result}.
\end{proof}

The inequality $ C_s(\rk) \leq C_M$ is intuitively clear because Alice cannot store a  message securely at a rate larger than the storage capacity of Bob's channel, which is $C_M$.  Assume that $\left(1-\frac{1}{\beta}\right)[q_0-p_0]^{+}+\rk \leq  C_M$. Then, the result of Theorem~\ref{thm:dna_wtc_key_capacity} can be interpreted as follows.  If the length of a shared key is $ML\rk$ bits, then Alice and Bob can use the initial part of the key of length 
$o(\log M)=o(L)$ for generating PCR primers. This will result in a DNA storage wiretap channel (with $q_0 > p_0$) over which the rest of the key of length $ML\rk - o(\log M)$ can be used for the purpose of wiretap channel coding. Since $\log M/ML \to 0$ as $M \to \infty$, the rate of the key that is being used for wiretap channel coding is still $\rk$. Hence, by Theorem~\ref{thm:dna_wtc_key_capacity}, the maximum rate at which Alice can securely store a message is $(1-1/\beta)(q_0-p_0)+\rk$.

On the other hand, if Alice and Bob were to use the entire shared key of length $ML\rk$ for one-time padding of a message, then the rate at which a message can be securely stored is $\rk$. Therefore, by using a negligible amount of the total key for creating a DNA storage wiretap channel, the users can improve the secure storage rates by $(1-1/\beta)(q_0-p_0)$ over that of the na\"ive one-time pad scheme.


\section{Proofs of Theorem~\ref{thm:dna_wtc_capacity} and Theorem~\ref{thm:dna_wtc_key_capacity}}\label{sec:proof_main_result}
Since the result of Theorem~\ref{thm:dna_wtc_capacity} can be obtained by setting the shared key $K$ to a constant, we have $C_s = C_s(\rk)$ when $\rk =0$.  Hence, it suffices to prove Theorem~\ref{thm:dna_wtc_key_capacity}, as Theorem~\ref{thm:dna_wtc_capacity} then follows immediately. In the rest of this section, we will focus on proving Theorem~\ref{thm:dna_wtc_key_capacity}.

\subsection{Achievability part}\label{sec:achievability}
Alice encodes a message by using distinct indices for each of the $M$ sequences. The initial segment of length $\log M$ of a sequence contains the index, and the rest of the sequence is used to encode the message. However, since the index is only used to order the molecules, the rate is scaled by a factor of $\left(1-\frac{1}{\beta}\right)$. The technique of indexing converts a noise-free shuffling-sampling channel with distribution $(\pi_0,\pi_1,\ldots)$ into a block-erasure channel (acting on a block of length $L-\log M$) with erasure probability $\pi_0$; see Appendix~\ref{app:block-erasure} for a detailed description of the block-erasure channel and block-erasure wiretap channel. It is important to note that a block-erasure channel is not memoryless. However, the channel is close to being memoryless: we can view each block as a symbol, and the channel acts independently on each of the blocks. If the size of the symbol alphabet were fixed (as $M \to \infty$), then we could directly apply standard results for discrete memoryless channels. However, the symbol alphabet $|\Sigma|^{L-\log M}$ grows with $M$. To handle this, we apply capacity results for ``information-stable'' channels from \cite{Dobrushin,verdu_han_capacity,vembu95} --- see Theorems~\ref{thm:capacity_stable} and~\ref{thm:capacity_bewtc} in Appendix~\ref{app:block-erasure}. In proving Theorem~\ref{thm:capacity_bewtc}, we use the fact that the blocklength of each sequence grows logarithmically in $M$ to show the strong secrecy condition \eqref{eq:eve_secrecy}.

As a result of indexing,  both Bob and Eve's channels are equivalent to block-erasure channels with respective erasure probabilities $\epsilon_{\text{m}}=p_0$ and $\epsilon_{\text{w}}=q_0$, yielding a block-erasure wiretap channel. 
Since a shared key of rate $\rk$ (which is $H(K)$ normalized by $ML$) is being used to encode into a block of length $M(L-\log M)$, the rate of the shared key used for the block-erasure channel is $R'_K:= \frac{1}{1-1/\beta}\rk$. Hence, by applying Theorem~\ref{thm:capacity_bewtc} with $m=M$, $l\sim (\beta-1)\log M$, and $R'_K:= \frac{1}{1-1/\beta}\rk$, we obtain
\begin{align*}
   C_s & \geq\left(1-\frac{1}{\beta}\right)\min\left\{R'_K+[\epsilon_{\text{w}}-\epsilon_{\text{m}}]^{+}, 1-\epsilon_{\text{m}}\right\}\\
   &= \left(1-\frac{1}{\beta}\right)\min\left\{\frac{1}{1-1/\beta}\rk +[q_0-p_0]^{+}, 1-p_0\right\}\\
   &= \min\left\{\left(1-\frac{1}{\beta}\right) [q_0-p_0]^{+} + \rk, \left(1-\frac{1}{\beta}\right)(1-p_0)\right\}.
\end{align*}

\subsection{Converse part}\label{sec:converse}
In this section, we prove the converse result that
\begin{align}\label{ineq:conv_proof}
    C_s \leq \min \left\lbrace\left(1-\frac{1}{\beta}\right)[q_0-p_0]^{+}+\rk,\  C_M\right\rbrace, 
\end{align}
where $C_M =\left(1-\frac{1}{\beta}\right)(1-p_0)$. Let us first argue that it is enough to prove \eqref{ineq:conv_proof} in the case of $q_0\geq p_0$.
Suppose $q_0 < p_0$. Consider a modified scenario  where Eve observes a processed version of $Z^{N_{\text{w}}L}$ instead. Here, ``processed version of $Z^{N_{\text{w}}L}$"  means that it is a random variable obtained by applying a stochastic function (independent of everything else) to $Z^{N_{\text{w}}L}$.  In particular, Eve observes $\tilde{Z}^{N'_{\text{w}}L}$ that is obtained by passing $Z^{N_{\text{w}}L}$ through a sampling channel with the distribution $(a_0, 1-a_0, 0, \ldots)$, where $a_0$ satisfies $\sum_{n=0}^{\infty}q_n a_0^n =p_0$.
We can view $\tilde{Z}^{N'_{\text{w}}L}$ as the output of a shuffling-sampling channel with distribution $(p_0, q'_1, \ldots)$ for some $q'_i$, $i \geq 1$. Therefore, we have a new DNA wiretap channel where Eve observes $\tilde{Z}^{N'_{\text{w}}L}$ while Bob observations remain the same. Since weakening Eve can only increase the secure storage capacity, $C_s$ is upper bounded by the secure storage capacity $\tilde{C}_s$ of the new DNA wiretap channel. Moreover, observe that for this new wiretap channel, the probability that Bob's channel does not sample a sequence is the same as Eve's channel. So, if we showed \eqref{ineq:conv_proof} in the case of  $q_0\geq p_0$, then it follows that
\begin{align}
    C_s \leq \tilde{C}_s \leq \left\lbrace \rk,\  \left(1-\frac{1}{\beta}\right)(1-p_0)\right\rbrace.
\end{align} 

Now consider the case of $q_0\geq p_0$. We prove \eqref{ineq:conv_proof} by considering a modified scenario. In this new scenario, Bob has access to a (genie-aided) side information $S$ in addition to $Y^{N_{\text{m}}L}$. As the side information can only increase the ability of Bob to recover the message, the capacity of this scenario is at least $C_s$. For Eve, we weaken her observations\footnote{If we provide Eve additional information, then the capacity will reduce. So, the idea of genie-aided Bob and genie-aided Eve will not work, in general, as the individual effects on capacity are conflicting with each other.} by providing her with a processed version of $Z^{N_{\text{w}}L}$. It is clear that weakening Eve can only make the secrecy capacity larger, so that 
\begin{align}
    C_s \leq \hat{C}_s,
    \label{ineq:Cs_le_hatCs}
\end{align} 
where $\hat{C}_s$ is the secure storage capacity of the model with a genie-aided Bob and a weaker Eve.

The side information $S$ we give to Bob is the same as that considered in the proof of \cite[Theorem~1]{shomorony21}:  If $Y_i^L$ and $Y_j^L$, $i\neq j$ are two identical $L$-length molecules, then $S$ distinguishes whether they were sampled from the same input molecule $X_k^L$ or from two identical input molecules $X^L_{k_1}$ and $X^L_{k_2}$, $k_1\neq k_2$. Using $(S,Y^{N_{\text{m}}L})$, Bob can compute the multiset $\hat{Y}^{M_{\text{m}}L} \subseteq \{X_1^L, \ldots,X_M^L\}$ that contains all the molecules of $X^{ML}$ that were sampled at least once. Here the random variable $M_{\text{m}}$ denotes the cardinality of the multiset. Let $\vfyh$ denote the frequency vector corresponding to $\hat{Y}^{M_{\text{m}}L}$ over $\Sigma^L$. Note  that the distribution of $(S, Y^{N_{\text{m}}L})$ depends on  $\vfx$ only through $\vfyh$, which implies the Markov chain
\begin{align}\label{mc:new_model:bob}
    (W, K) -\vfx-\vfyh-(S, Y^{N_{\text{m}}L}).
\end{align}

Instead of $Z^{N_{\text{w}}L}$, Eve has only access to $\hat{Z}^{M_{\text{w}}L}$, which is the set of all distinct $L$-length molecules in $\Sigma^L$ that appear in $Z^{N_{\text{w}}L}$. Let $\vfzh$ denote the frequency vector corresponding to $\hat{Z}^{M_{\text{w}}L}$ over $\Sigma^L$. The entries of the frequency vector are $\fzh(a_i)= \mathbbm{1}\{\fz(a_i)>0\}$ for $a_i \in \Sigma^L$, which indicates whether an $L$-length molecule appears in the multiset $Z^{N_{\text{w}}L}$ or not. By \eqref{mc:original_model}, we have 
\begin{align*}
    (W, K)-\vfx-\vfz-\vfzh.
\end{align*}
While there are other choices for Eve that provide a reasonable estimate for $X^{ML}$ through $Z^{N_{\text{w}}L}$ (maximum likelihood (ML) estimate of $\vfx$ based on $\vfz$ is one such choice), we choose $\hat{Z}^{M_{\text{w}}L}$ for the purpose of simpler analysis.

Let us derive an upper bound on the secure storage capacity $\hat{C}_s$ of the new scenario, where Bob has $(S, Y^{N_{\text{m}}L}, K)$ and Eve has  $\hat{Z}^{M_{\text{w}}L}$. Suppose that $R$ is an achievable secure message rate for the new scenario, i.e., there exists a sequence of pair of encoding and decoding functions $\{(\phi,\psi)\}_{M=1}^{\infty}$ that satisfy Bob's recoverability condition, $\mathbb{P}\{\psi(Y^{N_{\text{m}}L}, S, K) \neq W\} \to 0,$  and Eve's secrecy (strong) condition, $I(W ; \hat{Z}^{M_{\text{w}}L}) \to 0 $ as $M \to \infty$.  Then $R$ can be upper bounded as follows.
\begin{align}
         MLR\notag& = \log |\mathcal{W}| = H(W) =H(W|K)\notag\\
        & \stackrel{(a)}{\leq} H(W|K) - H(W \mid Y^{N_{\text{m}}L},S, K) + MLR\delta_M+1\notag\\
        & = I(W; Y^{N_{\text{m}}L},S|K) + MLR\delta_M+1\notag\\
        & = I(W,K; Y^{N_{\text{m}}L},S) + MLR\delta_M+1\label{ineq:conv_dna_channel_1}\\
        & = I(W; Y^{N_{\text{m}}L},S)+ I(W; K| Y^{N_{\text{m}}L},S) + MLR\delta_M+1\notag\\
        & \leq I(W; Y^{N_{\text{m}}L},S)+ H(K) + MLR\delta_M+1\notag\\
        & = I(W; Y^{N_{\text{m}}L},S)+ ML\rk + MLR\delta_M+1\notag\\
        & \stackrel{(b)}{\leq} I(W; Y^{N_{\text{m}}L},S) - I(W; 
        \hat{Z}^{N_{\text{w}}L}) + ML\rk + \delta'_M + MLR\delta_M+1\notag\\
        & \stackrel{(c)}{\leq} I(W; \vfyh) - I(W; \vfzh) + ML\rk + \delta'_M + MLR\delta_M+1\notag
\end{align}
where $(a)$ is because of the inequality $H(W \mid Y^{N_{\text{m}}L},S, K) \leq 1+MLR \delta_M$ for $\delta_M \to 0$, which is a consequence of Bob's recoverability condition \eqref{eq:bob_recoverability} and Fano's inequality,  $(b)$  follows from Eve's secrecy condition \eqref{eq:eve_secrecy}  where $\delta'_M \to 0$, and $(c)$ holds because of the data processing inequality  $I(W; Y^{N_{\text{m}}L}, S)\leq I(W; \vfyh)$  for the Markov chain \eqref{mc:new_model:bob}, and $I(W; \hat{Z}^{N_{\text{w}}L})= I(W;\vfzh)$. We can rewrite the above upper bound on $R$ as
\begin{align}\label{eq:mutual_difference}
   (1-\delta_M)R - \delta''_M - \rk &  \leq \frac{1}{ML} [I(W; \vfyh) - I(W; \vfzh)] 
\end{align}
where $\delta''_M=\frac{\delta'_M+1}{ML}\to 0$, $\delta_M\to 0$ and $W-\vfx-(\vfyh, \vfzh)$.\\

\subsubsection{Degradation of the frequency vector channels} The term $I(W; \vfyh) - I(W; \vfzh)$ in \eqref{eq:mutual_difference} depends on $P_{\vfyh,\vfzh|\vfx}$ only through the marginal distributions $P_{\vfyh|\vfx}$ and $P_{\vfzh|\vfx}$. Hence, with no loss in generality, we can work with a new coupled distribution $\tilde{P}_{\vfyh,\vfzh|\vfx}$ that has the same marginals as that of $P_{\vfyh,\vfzh|\vfx}$. For that, first note that the conditional joint distribution $P_{\vfyh,\vfzh|\vfx}$ is $P_{\vfyh|\vfx}P_{\vfzh|\vfx}$ because given $\vfx$, Bob's side information $S$  depends only on Bob's noise-free shuffling-sampling channel which is independent of Eve's channel.  Furthermore, we can decompose the channel between $\vfx$  and  $\vfyh ( \text{or }\vfzh)$ into  $|\Sigma^L|$ i.i.d. channels, one for each of the components of the frequency vectors, because the sampling process is independent across all the DNA molecules (see Fig.~\ref{fig:channel_decomposition}). Since $X^{ML}$ contains $M$ molecules, the input frequency vector is constrained by $\sum\limits_{i=1}^{|\Sigma^L|}\fx(a_i)=M$. Let $P_{\fyh\fzh|\fx}=P_{\fyh \mid \fx}P_{\fzh \mid \fx}$ denote the channel transition probability between the components $\fx(a_i)$ and $(\fyh(a_i), \fzh(a_i))$, for $a_i \in \Sigma^L$. Since $\fyh$ is a sum of $\fx$ number of independent Ber($p_0$) random variables and $\fzh=\mathbbm{1}\{\fz>0\}$, the channel transition probabilities are given by 
\begin{align*}
    P_{\fyh \mid \fx}(j \mid i)= \begin{cases}
    \binom{i}{j}p_0^{i-j}(1-p_0)^j, & \text{if $j \leq i$ and }\\
            0, & \text{otherwise}
    \end{cases}
\end{align*}
and 
\begin{align*}
    P_{\fzh \mid \fx}(j \mid i)= \begin{cases}
    q_0^i, & \text{if $j=0$ and}\\
            1-q_0^i, & \text{if $j=1$}
    \end{cases}
\end{align*}
where the alphabets for $\fx$, $\fyh$ and $\fzh$  are $\{0,1,\ldots,M\}$, $\{0,1,\ldots,M\}$ and $\{0,1\}$, respectively.

\begin{figure}[h]
\centering
\resizebox{0.9\width}{!}{\tikzstyle{block}=[rectangle, draw, thick, minimum width=2em, minimum height=3em]

\begin{tikzpicture}[node distance=3cm,auto,>=latex']
    \node (fx1) {$\fx(a_1)$};
    \node (ipdots) [below of=fx1, node distance= 1cm] {$\vdots$};
    \node (fxsigmaL) [below of=ipdots, node distance= 1.25cm] {$\fx(a_{|\Sigma^L|})$};
   
   \node [block, align=center] (ch1) [right of= fx1, node distance=3cm] {$P_{\fyh\fzh|\fx}$};
   \node (chdots) [right of= ipdots, node distance=3cm] {$\vdots$};
   \node [block, align=center] (chsigmaL) [right of= fxsigmaL, node distance=3cm] {$P_{\fyh\fzh|\fx}$};

    \node (fy1) [right of= ch1, node distance= 3cm, yshift=0.3cm]{$\fyh(a_1)$};
    \node (fz1) [right of= ch1, node distance= 3cm,yshift=-0.3cm]{$\fzh(a_1)$};
   \node (opdots) [right of= chdots, node distance=3cm] {$\vdots$};
    \node (fys) [right of= chsigmaL, node distance= 3cm, yshift=0.3cm]{$\fyh(a_{|\Sigma^L|})$};
    \node (fzs) [right of= chsigmaL, node distance= 3cm,yshift=-0.3cm]{$\fzh(a_{|\Sigma^L|})$};
   
    \draw[->, thick] (fx1.east) --  (ch1.west);
    \draw[->, thick] (fxsigmaL.east) --  (chsigmaL.west);
    \draw[->, thick] (fy1-|ch1.east) --  (fy1.west);
    \draw[->, thick] (fz1-|ch1.east) --  (fz1.west);
    \draw[->, thick] (fys-|chsigmaL.east) --  (fys.west);
    \draw[->, thick] (fzs-|chsigmaL.east) --  (fzs.west);
\end{tikzpicture}}
\caption{The channel between $\vfx$ and $(\vfyh,\vfzh)$ is decomposed into $|\Sigma^L|$ i.i.d. channels between the components with transition probability $P_{\fyh\fzh|\fx}$. The components of the input frequency vector $\vfx$ are constrained by $\sum\limits_{i=1}^{|\Sigma^L|}\fx(a_i)=M$.}
\label{fig:channel_decomposition}
 \end{figure}

\begin{lemma}\label{lem:degraded}
If $q_0\geq p_0$, then the channel $P_{\fzh \mid \fx}$ is a degraded version of $P_{\fyh \mid \fx}$, i.e., there exists a channel $Q$ such that
\begin{align*}
    P_{\fzh \mid \fx}(j|i)=\sum \limits_{k=0}^M Q(j|k)P_{\fyh \mid \fx}(k|i)
\end{align*}
for all $i \in \{0,1,\ldots,M\}$ and $j \in \{0,1\}$.
\end{lemma}
\begin{proof}
One can easily verify that the channel $Q$ defined by 
\begin{align*} 
    Q(j \mid k)=\begin{cases}
   \left(\frac{q_0-p_0}{1-p_0}\right)^k, & \text{ if } j=0 \\
            1-\left(\frac{q_0-p_0}{1-p_0}\right)^k, & \text{ if } j=1
    \end{cases}
\end{align*}
satisfies the degradation relation.
\end{proof}

By using Lemma~\ref{lem:degraded}, we consider the coupled distribution $\tilde{P}_{\fyh\fzh|\fx}=P_{\fyh \mid \fx}Q$ for the component channels with the marginals $P_{\fzh \mid \fx}$ and $P_{\fyh \mid \fx}$. This type of coupling for all the component channels yield a joint distribution that satisfy the Markov chain $W-\vfx-\vfyh-\vfzh$.  So,  we can write  $I(W; \vfyh) - I(W; \vfzh)=I(W; \vfyh\mid \vfzh) \leq I(\vfx; \vfyh\mid \vfzh)$. \\

\subsubsection{An upper bound in terms of the component channel mutual information}\label{sec:upp_component}
For $\sum\limits_{i =1}^{|\Sigma^L|}\fx(a_i)=M$, carrying on from \eqref{eq:mutual_difference}, we have 
\begin{align}
   (1 -\delta_M)&R - \delta''_M - \rk   \notag\\
   &\leq  \frac{1}{ML} [I(W; \vfyh) - I(W; \vfzh)] \notag \\
   & \leq \frac{1}{ML} I(\vfx; \vfyh\mid \vfzh)  \notag \\
   &= \frac{1}{ML}[H(\vfyh \mid \vfzh)-H(\vfyh \mid \vfx, \vfzh)]  \notag \\
   &= \frac{1}{ML}\sum\limits_{i=1}^{|\Sigma^L|}\left[ H(\fyh(a_i) \mid \vfzh, \fyh(a_1), \ldots \fyh(a_{i-1}))
   -H(\fyh(a_i) \mid \vfx,\vfzh, \fyh(a_1), \ldots \fyh(a_{i-1}))\right]  \notag \\
   & \stackrel{(a)}{\leq} \frac{1}{ML} \sum\limits_{i =1}^{|\Sigma^L|}\left[ H(\fyh(a_i) \mid \fzh(a_i))
   -H(\fyh(a_i)\mid \fx(a_i),\fzh(a_i))\right]\\
   &=\frac{1}{ML}\sum\limits_{i =1}^{|\Sigma^L|} I(\fx(a_i); \fyh(a_i)\mid \fzh(a_i)) \notag \\
   &=\frac{|\Sigma^L|}{ML} \sum\limits_{i =1}^{|\Sigma^L|} \frac{1}{|\Sigma^L|} I(\fx(a_i); \fyh(a_i) \mid \fzh(a_i)) \label{eq:convex_MI}
\end{align}
where $(a)$ follows from the fact that conditioning reduces entropy and that $\fyh(a_i)$ is conditionally independent of $\vfx,\vfzh, \fyh(a_1), \ldots \fyh(a_{i-1})$ given $\fx(a_i)$ and $\fzh(a_i)$. The summation in \eqref{eq:convex_MI} is a convex combination of conditional mutual information terms $I(\fx(a_i); \fyh(a_i) \mid \fzh(a_i))$ evaluated with respect to an input distribution $(m_0(i), \ldots, m_M(i)) \in \Delta^M$.  Let $\mu^F$ be the distribution induced by the encoder, $\phi$, on the set, $F$, of all input frequency vectors $\vf_X$ such that $\sum\limits_{i=1}^{|\Sigma^L|}\fx(a_i)=M$. Let $m_j(i)$ denote the probability (under $\mu^F$) that the component $\fx(a_i)$ equals $j$:
$$
m_j(i) = \sum_{\vfx \in F:\, \fx(a_i) = j} \mu^F(\vfx).
$$
Then, for any fixed $i$, we have 
$$
\sum_{j=0}^M j m_j(i) = \sum_{j=0}^M \sum_{\vfx: \fx(a_i) = j}j \mu^F(\vfx) = \sum_{\vfx \in F} \fx(a_i) \mu^F(\vfx),
$$
from which we obtain that $\sum\limits_{i=1}^{|\Sigma^L|}\sum\limits_{j=0}^{M}jm_j(i) =  \sum\limits_{\vfx \in F} \mu^F(\vfx) \sum\limits_{i=1}^{|\Sigma^L|} \fx(a_i) = M$. \\

We use the fact that for a Markov chain $X - Y - Z$, $I(X;Y|Z)$ is a concave function in the distribution of $X$ \cite[Lemma~1]{leung77} to conclude that
\begin{align}
   (1-\delta_M)R - \delta''_M  - \rk & \leq \frac{|\Sigma^L|}{ML} \sum\limits_{i =1}^{|\Sigma^L|} \frac{1}{|\Sigma^L|} I(\fx(a_i); \fyh(a_i) \mid \fzh(a_i)) \notag \\
   & \leq \frac{|\Sigma^L|}{ML} I(\fx; \fyh \mid \fzh) \label{ineq:IfX} \\
   &\leq \frac{|\Sigma^L|}{ML} \sup f(m_0, \ldots, m_M) \label{eq:supf}.
\end{align}
In \eqref{ineq:IfX}, $I(\fx; \fyh \mid \fzh)$ is evaluated at the input distribution $(m_0, \ldots, m_M)=\frac{1}{|\Sigma^L|}\Big(\sum \limits_{i=1}^{|\Sigma^L|}m_0(i), \ldots, \linebreak \sum \limits_{i=1}^{|\Sigma^L|}m_M(i)\Big)$, which satisfies the constraint $\sum \limits_{j=1}^{M}jm_j=\frac{M}{|\Sigma^L|}$, and the supremum of $f(m_0,\ldots,m_M) := I(\fx; \fyh \mid \fzh)$ in \eqref{eq:supf} is over such input distributions. By taking limits on both sides, we get the bound 
\begin{align}
    \hat{C}_s - \rk \leq \liminf_{M \to \infty}  \frac{|\Sigma^L|}{ML} \sup f(m_0,  \ldots, m_M)
    \label{Cs_upbnd}
\end{align}
where, again, the supremum is over $(m_0, \ldots, m_M) \in \Delta^M$ subject to $\sum \limits_{j=1}^{M}jm_j=\frac{M}{|\Sigma^L|}$.\\

\subsubsection{An upper bound on $f(m_0, \ldots, m_M)$ using constrained optimization}
To derive an upper bound on $\sup f(m_0,  \ldots, m_M)$, note that $f(m_0, \ldots, m_M) :=  I(\fx; \fyh \mid \fzh) = I(\fx; \fyh)- I(\fx; \fzh)$ evaluates to 
\begin{align}
H(l_0, l_1,\ldots, l_M) &- \sum \limits_{j=1}^M m_j H\big(\text{Bin}(j,p_0)\big)  -  h\Big(\sum\limits_{j=0}^M  m_j q_0^j\Big)+ \sum\limits_{j=1}^M m_j h(q_0^j),
\label{eq:fexpr}
\end{align}
 where $l_k=\sum \limits_{j\geq k}m_j \binom{j}{k}(1-p_0)^kp_0^{j-k}$. 
In Appendix~\ref{app:A}, we show that the first term above is 
\begin{equation}
H(l_0,l_1,\ldots,l_M) = h\Big(\sum \limits_{j=1}^M m_j(1-p_0^j)\Big) + o\left(\frac{ML}{|\Sigma^L|}\right)
\label{eq:term1}
\end{equation}
We drop the second term in \eqref{eq:fexpr}, and re-write the third term as $h\Big(\sum\limits_{i=0}^Mm_iq_0^i\Big)= h\Big(\sum \limits_{j=1}^M m_j(1-q_0^j)\Big)$.
The last term in \eqref{eq:fexpr} can be bounded as $\sum\limits_{j=1}^M m_j h(q_0^j) \leq \sum\limits_{j=1}^M m_j \leq \sum \limits_{j=1}^M jm_j= \frac{M}{|\Sigma^L|}$. Thus, for any $(m_0,m_1,\ldots,m_M) \in \Delta^M$ such that $\sum\limits_{j=1}^M j m_j = \frac{M}{|\Sigma^L|}$, where $M$ is sufficiently large enough, we have 
\begin{align}
    f(m_0, \ldots, m_M) &\leq h\Big(\sum \limits_{j=1}^M m_j(1-p_0^j)\Big)  - h\Big(\sum \limits_{j=1}^M m_j(1-q_0^j)\Big)+ o\left(\frac{ML}{|\Sigma^L|}\right) \nonumber\\
    & \leq  h\Big(\frac{M}{|\Sigma^L|}(1-p_0)\Big)-h\left(\frac{M}{|\Sigma^L|}(1-q_0)\right)+ o\left(\frac{ML}{|\Sigma^L|}\right) \label{f_upbnd}
\end{align}
the last inequality being proved in Appendix~\ref{app:B}.\\

\subsubsection{Evaluation of the upper bound on $\hat{C}_s$}
From \eqref{Cs_upbnd} and \eqref{f_upbnd}, we have
 \begin{align}
    \hat{C}_s  - \rk 
    &\leq \liminf_{M \to \infty}  \frac{|\Sigma^L|}{ML}\left[ h\left(\frac{M}{|\Sigma^L|}(1-p_0)\right)-h\left(\frac{M}{|\Sigma^L|}(1-q_0)\right)
    + o\left(\frac{ML}{|\Sigma^L|}\right) \right]\nonumber\\
    &= \liminf_{M \to \infty}  \frac{|\Sigma^L|}{ML} \left[h\left(\frac{M}{|\Sigma^L|}(1-p_0)\right)-h\left(\frac{M}{|\Sigma^L|}(1-q_0)\right)\right] \label{eq:binary_freq_term}
\end{align}
By setting $x := \frac{M}{|\Sigma^L|}$  and noting that $\frac{-\log x}{L}=\frac{L - \log M}{L} \to  1-\frac{1}{\beta}$ as $M \to \infty$, we obtain
\begin{align*}
    \hat{C}_s - \rk & \le \frac{\beta-1}{\beta}\liminf_{x \to 0} \frac{h((1-p_0)x)-h((1-q_0)x)}{-x\log x} = \left(1-\frac{1}{\beta}\right)(q_0-p_0),
\end{align*}
the evaluation of the limit requiring the use of L'H{\^o}pital's rule (twice). Combining this with \eqref{ineq:Cs_le_hatCs}, we have $C_s \le \left(1-\frac{1}{\beta}\right)(q_0-p_0)+\rk$.

To obtain the bound  $C_s \le \left(1-\frac{1}{\beta}\right)(1-p_0)$, we consider the inequality \eqref{ineq:conv_dna_channel_1}, which can be upper-bounded in a different way as
\begin{align*}
       MLR = I(W,K; Y^{N_{\text{m}}L},S) + MLR\delta_M+1
       &\leq I(W,K; \vfyh)+ MLR\delta_M+1\\
       & \leq H(\vfyh)+ MLR\delta_M+1.
\end{align*}
This inequality is nothing but  \cite[Eq.~(10)]{shomorony21}. It was shown in \cite{shomorony21} that\footnote{If we set $\vfzh$ to be a constant (which corresponds to setting $q_0=1$) in the analysis of the term $I(\vfx;\vfyh|\vfzh)$ in Sec.~\ref{sec:upp_component} and the subsequent sections, we will get the bound $\left(1-\frac{1}{\beta}\right)(1-p_0)$, which is identical to the one that we get from \cite{shomorony21}.}
\begin{align*}
    R \leq \left(1-\frac{1}{\beta}\right)(1-p_0),
\end{align*}
which yields
\begin{align*}
    C_s \leq \hat{C}_s  \leq \left(1-\frac{1}{\beta}\right)(1-p_0),
\end{align*}
completing the proof of the converse part of Theorem~\ref{thm:dna_wtc_capacity}.

\begin{remark} The converse proof relies on the idea of showing first that the encoding schemes which map messages to binary frequency vectors, i.e., there are no repeated molecules in $X^{ML}$, are optimal. The expression on the r.h.s. of  \eqref{eq:binary_freq_term} corresponds to this case. This can be seen by using the fact that the component channels $P_{\fyh \mid \fx}$ and $P_{\fzh \mid \fx}$  are just $Z$-channels with parameters $p_0$ and $q_0$, respectively. Note that index-based schemes are a special case of the encoding schemes involving binary frequency vectors. 
\end{remark}

\section{Discussion}\label{sec:discussion}
The DNA storage wiretap channel model considered in this paper is motivated by the fact that differential knowledge of primers creates a statistical advantage for the authorized party, Bob, over the unauthorized party, Eve. The take-away message from our work is that by exploiting this advantage using wiretap channel coding schemes, we can obtain information-theoretically secure DNA-based storage within our model. 

The crucial part of our characterization of the secure storage capacity of our model is the converse proof, which is based on analytically solving an optimization problem. An alternative proof, which is along the lines of that given for the converse part of Theorem~1 in \cite{shomorony21}, is provided in Appendix~\ref{app:C} for the special case when Eve's sampling distribution is Bernoulli $Q = (q_0,q_1)$. We were unable to find a way of extending this argument directly to more general distributions $Q$.

In future work, we intend to consider the DNA wiretap channel model with noisy shuffling-sampling channels as its components. An immediate direction would be to study the effect of substitution errors on the secure storage capacity. We also note that addressing insertion and deletions errors in the context of a DNA storage wiretap channel would be a challenging task as not much is known about their effect on the storage capacity of the shuffling-sampling channel itself.

\section*{Acknowledgements}
The authors would like to thank Dr.~Vamsi Veeramachaneni, Johns U.K., and Puspabeethi Samanta for helpful discussions during the course of this work.

\appendices
\section{Oligo Erasure Probabilities}
\label{app:erasure_probs}

In this appendix, we make some observations concerning the probability that a given oligo from a pool goes undetected after a sequencing run on a sample consisting of a composite pool of information-bearing oligos mixed with background DNA material. 

Consider a multiset $X^{ML} = \{X_1^L,\ldots,X_M^L\}$ of $M$ distinct oligos encoding some information of interest to us. After synthesis, amplification, and preparation of the sample for sequencing, assume that $X_i^{L}$ has an physical redundancy of $\kappa_i$, meaning that $\kappa_i$ identical copies of $X_i^L$ are present in the sample to be sequenced. We assume that these $\kappa_i$'s are i.i.d.\ random variables tightly concentrated about their mean $\overline{\kappa}$. Thus, $\overline{\kappa} M$ is, approximately, the total number of information-bearing DNA fragments in the sample. Let $N_B$ denote the total number of ``background'' DNA fragments added to the sample. Thus, the sample contains $N_{\textsc{tot}} := \overline{\kappa}M + N_B$ DNA fragments. The proportion of informative DNA in the sample is $\rho := \frac{\overline{\kappa}M}{N_{\textsc{tot}}}$.

This sample is fed into a sequencing platform that can generate $N_r$ reads. We will throughout assume that $N_r \ll N_{\textsc{tot}}$. 

For a fixed subset $S \subseteq [M]$, we are interested in the probability, $\theta_S$, that a sequencing run with $N_r$ reads fails to detect any of the oligos $X_i^L$, $i \in S$. This probability is well-estimated as follows: 
\begin{align}
\theta_S \approx \frac{\binom{\overline{\kappa}(M-|S|)+N_B}{N_r}}{\binom{N_{\textsc{tot}}}{N_r}} = \frac{\binom{N_{\textsc{tot}} - \overline{\kappa}|S|}{N_r}}{\binom{N_{\textsc{tot}}}{N_r}} =\frac{\frac{(N_{\textsc{tot}}-\overline{\kappa}|S|)!}{(N_{\textsc{tot}}-\overline{\kappa}|S|-N_r)!}}{\frac{N_{\textsc{tot}}!}{(N_{\textsc{tot}}-N_r)!}} &= \frac{\prod_{j=0}^{N_r-1}  (N_{\textsc{tot}}- \overline{\kappa}|S| - j)}{\prod_{j=0}^{N_r-1}  (N_{\textsc{tot}} - j)}  \notag \\
&= \prod_{j=0}^{N_r-1} \left(1-\frac{\overline{\kappa}|S|}{N_{\textsc{tot}}-j}\right).
\label{eq:theta_S}
\end{align}

%

\subsection*{Case 1: $\rho \ll 1$} We consider here $\rho$ to be $10^{-3}$ or less. This is the case assumed to be encountered by Eve. In this case, since $\overline{\kappa}M \ll N_{\textsc{tot}}$ and $N_r \ll N_{\textsc{tot}}$, we can use the approximation $1-\frac{\overline{\kappa}m}{N_{\textsc{tot}}-j} \approx \exp\bigl(-\frac{\overline{\kappa}m}{N_{\textsc{tot}}}\bigr)$ for $0 \le m \le M$ and $0 \le j \le N_r-1$. Hence, for any $S \subseteq [M]$
\begin{align*}
\theta_S \approx \exp\biggl(- \frac{\overline{\kappa}|S| N_r}{N_{\textsc{tot}}} \biggr) = {[\exp(-\rho N_r /M)]}^{|S|},
\end{align*}
It follows that the number of oligos, out of the $M$ in the multiset $X^{ML}$, that go undetected after a sequencing run is well-approximated as a Bin$(M,q_0)$ random variable, where $q_0 := \exp(-\rho N_r/M)$ is the \emph{oligo erasure probability}.  

We want to keep a high oligo erasure probability at Eve's end, so a good choice of $\rho$ would be $0.5 M/N_r$, which would yield $q_0 \approx \exp(-1/2) \approx 0.6065$. For example, if there are $M = 10,000$ informative oligos, and Eve is allowed $N_r = 500$ million reads, then the sample for sequencing should be prepared with $\rho = 10^{-5}$. Doing so would ensure that the (random) number of informative oligos that are missed by Eve is tightly concentrated about $q_0 M \approx 6000$.

\subsection*{Case 2: $\rho$ ``not too small''} Here, we consider $\rho$ to be $0.05$ or more. This is the case assumed to be seen by Bob, since Bob is expected to run several cycles of PCR on the sample to amplify the informative DNA alone. Now, we may no longer have $\overline{\kappa}M \ll N_{\textsc{tot}}$, but the assumption that $N_r \ll N_{\textsc{tot}}$ is still valid, so we use the approximation $1-\frac{\overline{\kappa}m}{N_{\textsc{tot}}-j} \approx 1-\frac{\overline{\kappa}m}{N_{\textsc{tot}}} = 1- \rho m/M$, for $0 \le j \le N_r-1$ and $0 \le m \le M$. Hence, for any $S \subseteq [M]$, we have via \eqref{eq:theta_S},
$$
\theta_S \approx (1-\rho|S|/M)^{N_r}.
$$

In particular, the probability that any given single oligo (out of the $M$ in the multiset $X^{ML}$) is undetected after a sequencing run is approximately $p_0 := (1-\rho/M)^{N_r}$, which we may take to be Bob's oligo erasure probability.

 For example, with $\rho=0.1$ and $M=10,000$, if Bob is allowed $N_r = 1$ million reads, then $p_0 \approx 4.5 \times 10^{-5}$. So, out of the pool of $M=10,000$ informative oligos, Bob can expect to miss about $p_0  M$ oligos, which is less than half an oligo.

\section{Proof of Equation~\eqref{eq:term1}}
\label{app:A}
We begin by breaking up $H(l_0,l_1,\ldots,l_M)$ using the grouping property of entropy:
\begin{align}
H (l_0,l_1,\ldots,l_M) =\ H\Big(l_0,\sum_{i=1}^M l_i\Big)+\Big(\sum_{i=1}^M l_i\Big) H(\tilde{l}_1, \ldots, \tilde{l}_M), \label{eq:H_grouping}
\end{align} 
where $\tilde{l}_i= \left.l_i\middle/\sum_{j=1}^M l_j\right.$, $1\leq i\leq M$. Note first that 
\begin{equation*}
    H\Big(l_0,\sum_{i=1}^M l_i\Big) \ = \ h\Big(\sum_{j=1}^M m_j(1-p_0^j)\Big),
    \label{eq:Hl0}
\end{equation*}
since 
\begin{align*}
    \sum_{i=1}^M l_i \ &= \ \sum_{i=1}^M \sum_{j\geq i}m_j \binom{j}{i}(1-p_0)^ip_0^{j-i} = \ \sum_{j=1}^M m_j \sum_{1\leq i\leq j}\binom{j}{i}(1-p_0)^ip_0^{j-i} = \ \sum_{j=1}^M m_j(1-p_0^j).
\end{align*}

It remains to show that the second term in \eqref{eq:H_grouping} is $o\left(\frac{ML}{|\Sigma^L|}\right)$. To this end, define $$g^{\#}(M):=\max  \Big(\sum_{i=1}^M l_i\Big) H(\tilde{l}_1, \ldots, \tilde{l}_M),$$
where the maximum is over $(m_0,\ldots,m_M)\in \Delta^M$ such that $\sum_{j=1}^{M} jm_j=\frac{M}{|\Sigma^L|}$. The aim is to show that $g^{\#}(M)=o\left(\frac{ML}{|\Sigma^L|}\right)$.

Recall that $l_i:=\sum_{j\geq i}m_j \binom{j}{i}(1-p_0)^ip_0^{j-i}$ for $1\leq i \leq M$. Consider the sum $\sum_{i=1}^Mil_i=\linebreak\sum_{i=1}^Mi\sum_{j\geq i}m_j \binom{j}{i}(1-p_0)^ip_0^{j-i} =\sum_{j=1}^Mm_j\sum_{1\leq i\leq j}i \binom{j}{i}(1-p_0)^ip_0^{j-i}=\sum_{j=1}^Mm_j j(1-p_0)=(1-p_0)\sum_{j=1}^M jm_j$. So the constraint $\sum_{i=1}^{M}im_i=\frac{M}{|\Sigma^L|}$ is equivalent to the constraint $\frac{1}{1-p_0}\sum_{i=1}^Mil_i=\frac{M}{|\Sigma^L|}$.
Since $1-p_0\leq 1$, we have $il_i\leq \frac{il_i}{1-p_0}$. Therefore, the constraint $\sum_{i=1}^{M}\frac{il_i}{1-p_0}=\frac{M}{|\Sigma^L|}$ implies that $\sum_{i=1}^{M}il_i\leq \frac{M}{|\Sigma^L|}$. With this constraint relaxation, we can upper bound $g^{\#}(M)$ as
\begin{align}
    g^{\#}(M) &\leq \max \Big(\sum_{j=1}^M a_j\Big) H(\tilde{a}_1, \ldots, \tilde{a}_M)
    \label{eq:g_upbnd}
\end{align}
where $\tilde{a}_i= \left.a_i\middle/\sum_{j=1}^M a_j\right.$ for $1\leq i\leq M$ and the maximum is over $(a_1,\ldots,a_M)\geq \boldsymbol{0}$ such that $\sum_{i=1}^{M}ia_i\leq \frac{M}{|\Sigma^L|}$. Note that
\begin{align*}
   \Big(\sum_{j=1}^M a_j\Big) & H(\tilde{a}_1, \ldots, \tilde{a}_M) = \ - \sum_{j=1}^M a_j \log \tilde{a}_j =\Big(\sum_{j=1}^M a_j\Big) \log\Big(\sum_{j=1}^M a_j\Big) - \sum_{j=1}^M a_j \log a_j. 
\end{align*}
 This is a concave function in $(a_1,\ldots,a_M)$, which can be seen by applying the log-sum inequality: for $\gamma \in [0,1]$, $(a_1,\ldots,a_M)\geq \boldsymbol{0}$ and $(b_1,\ldots,b_M)\geq \boldsymbol{0}$, we have 
\begin{align*}
    \gamma a_j \log \frac{\gamma a_j}{\gamma \sum_{j=1}^M a_j} + (1-\gamma) b_j \log \frac{(1-\gamma) b_j}{(1-\gamma) \sum_{j=1}^M b_j} \geq (\gamma a_j+(1-\gamma) b_j) \log \frac{\gamma a_j+(1-\gamma) b_j}{\gamma \sum_{j=1}^M a_j+(1-\gamma) \sum_{j=1}^M b_j},
\end{align*}
and the concavity follows after multiplying both sides by $-1$ and taking summation. 

We can assume that the maximizer $(a^*_1,\ldots,a^*_M)$ is such that $a_j^* > 0$ for all $j$. If not, then there exist two distinct indices $i_1<i_2$ and $j$ so that $a_j^*=0$ and $a_{i_1}^*\neq 0 \neq a_{i_2}^*$ (since the maximum in \eqref{eq:g_upbnd} is positive, there must be at least two non-zero entries in the maximizer). Set $b_{i_1}=a_{i_1}^*-\epsilon\left(\frac{j}{2i_1}\right)$, $b_{i_2}=a_{i_2}^*-\epsilon\left(\frac{j}{2i_2}\right)$, $b_{j}=\epsilon$ and $b_l=a_l^*$ for $l \notin \{i_1,i_2,j\}$, where $\epsilon$ is small enough that $b_{i_1},b_{i_2}>0$. Note that for this choice $\sum_{i=1}^{M}ia^*_i=\sum_{i=1}^{M}ib_i$. Since $\frac{d}{d\epsilon}\left[\Big(\sum_{j=1}^M b_j\Big) \log\Big(\sum_{j=1}^M b_j\Big) - \sum_{j=1}^M b_j \log b_j\right] \to \infty$ as $\epsilon \to 0^+$, there exists another point in the neighbourhood of $(a^*_1,\ldots,a^*_M)$ that takes a strictly larger value, hence violating the optimality of $(a^*_1,\ldots,a^*_M)$. 

The KKT conditions corresponding to the optimization problem in \eqref{eq:g_upbnd} are
\begin{align*}
   &\log\Big(\sum_{j=1}^M a_j\Big) -\log a_j -j\lambda =0, \quad j=1, \ldots, M,\\
   &\lambda \bigg( \frac{M}{|\Sigma^L|}-\sum_{j=1}^{M}ja_j \bigg)=0,\\
   &(a_1,\ldots,a_M)> \boldsymbol{0}, \quad  \sum_{j=1}^{M}ja_j\leq \frac{M}{|\Sigma^L|} \text{ and } \lambda \geq 0.
\end{align*}

Suppose $\lambda =0$, then $\log\Big(\sum_{j=1}^M a_j\Big) = \log a_j$ for all $j$, which forces all the entries to zero and the corresponding  value of the function is zero (not the maximum value). When $\lambda>0$, $\sum_{j=1}^{M}ja_j= \frac{M}{|\Sigma^L|}$. This implies that $\frac{a_j}{\sum_{j=1}^M a_j}= 2^{-\lambda j}$, for all $j$. Taking summation on both the sides, we get the condition $\sum_{j=1}^M 2^{-\lambda j}=1$. Let $C:=\sum_{j=1}^M a_j$, which satisfies  $C\sum_{j=1}^{M}j2^{-\lambda j}= \frac{M}{|\Sigma^L|}$. Since the function is concave and all the coordinates of the function are positive, the KKT necessary conditions give the optimality of this point. Hence the optimizer is  $(a^*_1,\ldots,a^*_M)=(C2^{-\lambda}, \ldots,C2^{-M\lambda})$ where $\lambda$ is chosen such that $\sum_{j=1}^M 2^{-\lambda j}=1$ and $C= \frac{M}{|\Sigma^L|\sum_{j=1}^{M}j2^{-\lambda j}}$. As $1=\sum_{j=1}^M 2^{-\lambda j}\leq \sum_{j=1}^{\infty} 2^{-\lambda j}=\frac{2^{-\lambda}}{1-2^{-\lambda}}$,  we have $\lambda \leq 1$. By evaluating the function at this optimizer, we obtain
\begin{align*}
    g^{\#}(M) \leq \max \Big(\sum_{j=1}^M a_j\Big) H(\tilde{a}_1, \ldots, \tilde{a}_M) &=\Big(\sum_{j=1}^M a^*_j\Big) \log\Big(\sum_{j=1}^M a^*_j\Big) - \sum_{j=1}^M a^*_j \log a^*_j\\
    &=C\log C - \sum_{j=1}^M\left( C 2^{-\lambda j} \log C 2^{-\lambda j}\right)\\
    &=C\log C - \sum_{j=1}^M C 2^{-\lambda j} \log C - \sum_{j=1}^M C 2^{-\lambda j} \log 2^{-\lambda j}\\
    &=C\log C-C\log C + C\lambda \sum_{j=1}^M j 2^{-\lambda j}\\
    &=C\lambda \frac{M}{|\Sigma^L|C} \leq  \frac{M}{|\Sigma^L|} 
\end{align*}
Thus we have shown that $g^{\#}(M)=o\left(\frac{ML}{|\Sigma^L|}\right)$, as desired.

\section{Proof of Inequality~\eqref{f_upbnd}}
\label{app:B}
We formally state as a proposition the statement we wish to prove.

\begin{proposition} \label{prop:sup_entropy_diff}
Let $\delta:=\frac{M}{|\Sigma^L|}$ and  $0 \leq p_0\leq q_0\leq 1$. For sufficiently large $M$,
\begin{align}\label{eq:sup_entropy_diff}
    \max \left[h\biggl(\sum_{j=1}^M m_j(1-p_0^j)\biggr) - h\biggl(\sum_{j=1}^M m_j(1-q_0^j)\biggr) \right] =  h(\delta(1-p_0))-h(\delta(1-q_0))
\end{align}
where the maximum is over  $(m_0,\ldots,m_M)\in \Delta^M$ such that $\sum_{i=1}^{M}im_i=\delta$.
\end{proposition}

\begin{proof}
First, note that 
\begin{align*}
\sup \left[h\biggl(\sum_{j=1}^M m_j(1-p_0^j)\biggr) - h\biggl(\sum_{j=1}^M m_j(1-q_0^j)\biggr) \right]  \geq  h(\delta(1-p_0))-h(\delta(1-q_0))
\end{align*}
which is obtained by setting $(m_0,m_1,\ldots,m_M)=(1-\delta,\delta, \ldots,0)$.

For the other direction, let  $\lambda_i:= \frac{im_i}{\delta}$ for $1\leq i\leq M$. With this change of variables, the set of  constraints, $(m_0,\ldots,m_M)\in \Delta^M$ and $\sum_{i=1}^{M}im_i=\delta$,  imply the constraint $(\lambda_1,\ldots, \lambda_M) \in \Delta^{M-1}$. Hence, the left-hand side of \eqref{eq:sup_entropy_diff} is bounded above by $\sup \bigg[h\Big(\delta\sum_{j=1}^M \lambda_j\frac{(1-p_0^j)}{j}\Big) - h\Big(\delta\sum_{j=1}^M \lambda_j\frac{(1-q_0^j)}{j}\Big) \bigg]$ over $\Delta^{M-1}$.  Since the function is continuous on the compact set $\Delta^{M-1}$, the supremum is attained at a point in $\Delta^{M-1}$. In particular, we show that the maximum value is attained at the vertex $(1,0,\ldots,0)$ for sufficiently small $\delta$ (or equivalently, for sufficiently large $M$). In other words, we prove that for sufficiently small $\delta$, 
\begin{align}
    h\biggl(\delta\sum_{j=1}^M \lambda_j\frac{(1-p_0^j)}{j}\biggr) - h\biggl(\delta\sum_{j=1}^M \lambda_j\frac{(1-q_0^j)}{j}\biggr)   \leq h(\delta(1-p_0))-h(\delta(1-q_0))
    \label{eq:sup_entropy_diff_bound}
\end{align}
for every $(\lambda_1,\ldots, \lambda_M) \in \Delta^{M-1}$.

Let us first consider the $q_0=1$ case. The inequality \eqref{eq:sup_entropy_diff_bound} becomes $h\biggl(\delta\sum_{j=1}^M \lambda_j\frac{(1-p_0^j)}{j}\biggr) \leq h(\delta(1-p_0))$. Note that  $w(x):=1-p_0^x$ is a concave function, and $w(0)=0$; both these conditions imply the relation $\frac{(1-p_0^{j_1})}{j_1} \geq \frac{(1-p_0^{j_2})}{j_2}$ if $j_1 \leq j_2$. As a result, we have $\sum_{j=1}^M \lambda_j\frac{(1-p_0^j)}{j} \leq (1-p_0)\sum_{j=1}^M \lambda_j=1-p_0$. If we choose $\delta$ such that $\delta(1-p_0)\leq \frac{1}{2}$, then we can conclude  that $h\biggl(\delta\sum_{j=1}^M \lambda_j\frac{(1-p_0^j)}{j}\biggr) \leq h(\delta(1-p_0))$ because $h(x)$ is monotonically increasing in the interval $\left[0, \frac{1}{2}\right]$.

Now consider the case of $q_0<1$. Define $u(x)=h\biggl(\delta\sum_{j=1}^M \lambda_j\frac{(1-x^j)}{j}\biggr)$ and $v(x)=h(\delta(1-x))$. The derivatives of these functions are $u'(x)=-\delta\Big(\sum_{j=1}^M \lambda_jx^{j-1}\Big)h'\biggl(\delta\sum_{j=1}^M \lambda_j\frac{(1-x^j)}{j}\biggr)$ and $v'(x)=-\delta h'(\delta(1-x))$, where $h'(y)= \log \frac{1-y}{y}$ for $y \in (0,1)$. By Lemma~\ref{lem:entropy_derv_bound} below, if we choose $\delta$ sufficiently small, then $u'(x)\geq v'(x)$ for every $x \in [0,q_0]$ and every $(\lambda_1,\ldots, \lambda_M) \in \Delta^{M-1}$. So, by using the fundamental theorem of calculus and the monotonicity of integrals, we obtain
\begin{align*}
    h\biggl(\delta\sum_{j=1}^M \lambda_j\frac{(1-p_0^j)}{j}\biggr) - h\biggl(\delta\sum_{j=1}^M \lambda_j\frac{(1-q_0^j)}{j}\biggr) \ = \ -\int_{p_0}^{q_0}u'(x) dx &\ \leq \ -\int_{p_0}^{q_0}v'(x) dx  \\ & =h(\delta(1-p_0))-h(\delta(1-q_0)).
\end{align*}
This completes the proof.
\end{proof}

\begin{lemma}\label{lem:entropy_derv_bound}
Let $q_0<1$ be fixed. If $\delta > 0$ is sufficiently small, then for every $x \in [0,q_0]$ and  $(\lambda_1,\ldots, \lambda_M) \in \Delta^{M-1}$, we have
$$
\Big(\sum_{j=1}^M \lambda_jx^{j-1}\Big)h'\biggl(\delta\sum_{j=1}^M \lambda_j\frac{(1-x^j)}{j}\biggr) \leq h'(\delta(1-x)).
$$
\end{lemma}
\begin{proof}
Fix an $x \in [0,q_0]$, and let $\psi(\lambda_1,\ldots, \lambda_M):= \Big(\sum_{j=1}^M \lambda_jx^{j-1}\Big)h'\biggl(\delta\sum_{j=1}^M \lambda_j\frac{(1-x^j)}{j}\biggr)$. Note that $\psi(1,\ldots, 0)=h'(\delta(1-x))$.  Since $\lambda_M=1-\sum_{j=1}^{M-1} \lambda_j$, $\phi(\lambda_1,\ldots, \lambda_{M-1}) := \psi (\lambda_1,\ldots, \lambda_{M-1},1-\sum_{j=1}^{M-1} \lambda_j)$  is a function of $M-1$ variables $(\lambda_1,\ldots, \lambda_{M-1})\in \Lambda:=\Big\{ (\lambda_1,\ldots, \lambda_{M-1}) \geq \boldsymbol{0}: \sum_{j=1}^{M-1} \lambda_j \leq 1\Big\}$. Note that a point $(\mu_1,\ldots, \mu_{M-1})$ is a local maximizer of $\phi$ if for all feasible directions $(d_1,\ldots, d_{M-1})$, $\sum_{j=1}^{M-1}\frac{\partial \phi}{\partial \lambda_j}d_j \leq 0$. The goal is to show that a point $(\mu_1,\ldots, \mu_{M-1}) \neq (1,0,\ldots,0)$ is not a local maximizer, i.e., there exists a feasible direction $(d_1,\ldots, d_{M-1})$ such that $\sum_{j=1}^{M-1}\frac{\partial \phi}{\partial \lambda_j}d_j > 0$. Moreover, we show that at the point $(1,0,\ldots,0)$, $\sum_{j=1}^{M-1}\frac{\partial \phi}{\partial \lambda_j}d_j < 0$  for all feasible directions $(d_1,\ldots, d_{M-1})$, which implies that it is a local maximizer. Hence, $(1,0,\ldots,0)$ is the global maximizer.

It is sufficient to show that $\frac{\partial \phi}{\partial \lambda_1} > 0$ and
$\min_{2\leq j \leq M-1}\left\{\frac{\partial \phi}{\partial \lambda_1}-\frac{\partial \phi}{\partial \lambda_j}\right\}>0$ at all points of $\Lambda$ and for all $x \in [0,q_0]$ because of the following observation. For a point $(\mu_1,\ldots, \mu_{M-1}) \neq (1,0,\ldots,0)$ in $\Lambda$, if we choose the direction $(d_1,\ldots, d_{M-1})=(1-\mu_1,0-\mu_2,\ldots,0-\mu_{M-1})$, then $\sum_{j=1}^{M-1}\frac{\partial \phi}{\partial \lambda_j}d_j =  \frac{\partial \phi}{\partial \lambda_1}(1-\mu_1)- \sum\limits_{j=2}^{M-1}\frac{\partial \phi}{\partial \lambda_j}\mu_j \geq \frac{\partial \phi}{\partial \lambda_1}\Big(\sum\limits_{j=2}^{M-1} \mu_j\Big)- \sum\limits_{j=2}^{M-1}\frac{\partial \phi}{\partial \lambda_j}\mu_j = \sum\limits_{j=2}^{M-1}\left( \frac{\partial \phi}{\partial \lambda_1}-\frac{\partial \phi}{\partial \lambda_j}\right)\mu_j\geq 0$. Observe that at least one of the inequalities is strict. If $\sum_j \mu_j < 1$, then the first inequality is strict. In the case of $\sum_j \mu_j = 1$, the condition $\mu_1 \neq 1$ implies that there exists an index $j \neq 1$ such that $\mu_j >0$. As a result,  the last inequality is strict. At the point $(1,0,\ldots,0)$, the feasible directions are $(d_1,\ldots, d_{M-1})=(\mu_1-1,\mu_2-0,\ldots,\mu_{M-1}-0)$ for every $(\mu_1,\ldots, \mu_{M-1}) \in \Lambda \setminus \{(1,0,\ldots,0)\}$. So,  we have $\sum_{j=1}^{M-1}\frac{\partial \phi}{\partial \lambda_j}d_j=  \frac{\partial \phi}{\partial \lambda_1}(\mu_1-1)+ \sum\limits_{j=2}^{M-1}\frac{\partial \phi}{\partial \lambda_j}\mu_j \leq  -\frac{\partial \phi}{\partial \lambda_1}\Big(\sum\limits_{j=2}^{M-1} \mu_j\Big)+ \sum\limits_{j=2}^{M-1}\frac{\partial \phi}{\partial \lambda_j}\mu_j = \sum\limits_{j=2}^{M-1}\left( \frac{\partial \phi}{\partial \lambda_j}-\frac{\partial \phi}{\partial \lambda_1}\right)\mu_j\leq 0$, where the strictness of one of the inequalities follows from the above argument. 

For $1\leq i \leq M-1$, the partial derivatives are \begin{align*}
    \frac{\partial \phi}{\partial \lambda_{i}}&=\left(x^{i-1}-x^{M-1}\right)h'\biggl(\delta\sum \limits_{j=1}^M \lambda_j\frac{(1-x^j)}{j}\biggr)\nonumber\\& \mkern 100mu + \delta \Big(\sum \limits_{j=1}^M \lambda_jx^{j-1}\Big)\left(\frac{(1-x^i)}{i}-\frac{(1-x^M)}{M}\right) h''\biggl(\delta\sum \limits_{j=1}^M \lambda_j\frac{(1-x^j)}{j}\biggr)\\&= \left(x^{i-1}-x^{M-1}\right)h'\biggl(\delta\sum \limits_{j=1}^M \lambda_j\frac{(1-x^j)}{j}\biggr) - \frac{\log_2e \cdot \Big(\sum \limits_{j=1}^M \lambda_jx^{j-1}\Big)\left(\frac{(1-x^i)}{i}-\frac{(1-x^M)}{M}\right) }{\biggl(\sum \limits_{j=1}^M \lambda_j\frac{(1-x^j)}{j}\biggr)\biggl(1-\delta\sum \limits_{j=1}^M \lambda_j\frac{(1-x^j)}{j}\biggr)}
\end{align*}
where $\lambda_M=1-\sum_{j=1}^{M-1} \lambda_j$, $h''(y)=-\frac{\log_2e}{y(1-y)}$ for $y\in (0,1)$. 
Since $\delta\sum_{j=1}^M \lambda_j\frac{(1-x^j)}{j} \leq \delta$ and $h'(y)$ is a monotonically decreasing non-negative function in $\left(0,\frac{1}{2}\right]$,  $h'\biggl(\delta\sum_{j=1}^M \lambda_j\frac{(1-x^j)}{j}\biggr) \geq h'(\delta)$. 
Note that the function  $r(y):=\frac{1}{a^y}-1$, for some $a\in (0,1]$, is a convex function with $r(0)=0$, therefore we have
$\frac{\frac{1}{a^{j_1}}-1}{j_1} \geq \frac{\frac{1}{a^{j_2}}-1}{j_2}$ if $j_1 \geq j_2$, which implies that $\frac{j_2a^{j_2-1}}{1-a^{j_2}}\geq \frac{j_1a^{j_1-1}}{1-a^{j_1}}$. So, $\max_{1\leq j\leq M} \frac{jx^{j-1}}{1-x^j} = \frac{1}{1-x}$ (which holds even at $x=0$). Therefore, we have
\begin{align*}
    \frac{\sum \limits_{j=1}^M \lambda_jx^{j-1}}{\sum \limits_{j=1}^M \lambda_j\frac{(1-x^j)}{j}} \leq \max \limits_{1\leq j\leq M} \frac{x^{j-1}}{\frac{(1-x^j)}{j}}= \max \limits_{1\leq j\leq M} \frac{jx^{j-1}}{1-x^j}=\frac{1}{1-x} \leq \frac{1}{1-q_0}.
\end{align*}
Furthermore, we have the bound 
\begin{align}
    \frac{\frac{(1-x^i)}{i}-\frac{(1-x^M)}{M} }{1-\delta\sum \limits_{j=1}^M \lambda_j\frac{(1-x^j)}{j}} \leq \frac{1}{1-\delta}.
\end{align}
By using the above bounds, the required quantities can be uniformly  bounded (not depending on $x$ and $(\lambda_1,\ldots, \lambda_{M-1})$) as
\begin{align*}
     \frac{\partial \phi}{\partial \lambda_{1}}&= \left(1-x^{M-1}\right)h'\biggl(\delta\sum \limits_{j=1}^M \lambda_j\frac{(1-x^j)}{j}\biggr) - \frac{\log_2e \cdot \Big(\sum \limits_{j=1}^M \lambda_jx^{j-1}\Big)\left((1-x)-\frac{(1-x^M)}{M}\right) }{\biggl(\sum \limits_{j=1}^M \lambda_j\frac{(1-x^j)}{j}\biggr)\biggl(1-\delta\sum \limits_{j=1}^M \lambda_j\frac{(1-x^j)}{j}\biggr)}\\
     &\geq (1-q_0)h'(\delta)-\frac{\log_2e}{(1-q_0)(1-\delta)}
\end{align*}
and
\begin{align*}
     \frac{\partial \phi}{\partial \lambda_{1}}-\frac{\partial \phi}{\partial \lambda_{i}}&= \left(1-x^{i-1}\right)h'\biggl(\delta\sum \limits_{j=1}^M \lambda_j\frac{(1-x^j)}{j}\biggr) - \frac{\log_2e \cdot \Big(\sum \limits_{j=1}^M \lambda_jx^{j-1}\Big)\left((1-x)-\frac{(1-x^i)}{i}\right) }{\biggl(\sum \limits_{j=1}^M \lambda_j\frac{(1-x^j)}{j}\biggr)\biggl(1-\delta\sum \limits_{j=1}^M \lambda_j\frac{(1-x^j)}{j}\biggr)}\\
     &\geq (1-q_0)h'(\delta)-\frac{\log_2e}{(1-q_0)(1-\delta)}
\end{align*}
for $2\leq i \leq M-1$. It is enough to choose a small $\delta$ such that $(1-q_0)h'(\delta)-\frac{\log_2e}{(1-q_0)(1-\delta)}>0$, which is possible as $\delta \to 0$, $h'(\delta) \to \infty$ and $\frac{1}{1-\delta} \to 1$. 
\end{proof}

\section{A simple converse result of Theorem~\ref{thm:dna_wtc_capacity} when $Q=(q_0,q_1)$}\label{app:C}
In the special case when Eve's sampling distribution is $Q=(q_0,q_1)$,  we can prove the converse result by following a simple combinatorial approach similar to that of Shomorony and Heckel \cite{shomorony21}. As in the converse proof of Theorem~\ref{thm:dna_wtc_capacity} given in Section~\ref{sec:converse}, we provide Bob side information  which allows him to tell whether or not two identical molecules in his sample $Y^{N_{\text{m}}L}$ came from the same input molecule $X_k^L$. Note that for Eve, this ambiguity does not arise, since her sampling distribution $Q$ has $q_j = 0$ for all $j \geq 2$. It is for this reason that the simple combinatorial argument given here works. Let $\vfyh$ be the frequency vector available to Bob (as described in Section~\ref{sec:converse}) and let $\vfz$ denote the frequency vector corresponding to the multiset $Z^{N_{\text{w}}L}$ observed by Eve.

Our starting point for the converse proof is again \eqref{eq:mutual_difference}:
\begin{align*}
   (1-\delta_M)R - \delta''_M - \rk &  \leq \frac{1}{ML} [I(W; \vfyh) - I(W; \vfz)].
\end{align*}
As in the argument given in Section~\ref{sec:converse}, we can consider a new coupled distribution $\tilde{P}_{\vfyh,\vfz|\vfx}$ that has the same marginals as that of $P_{\vfyh,\vfz|\vfx}$. This new distribution is constructed by decomposing the channel between $\vfx$  and  $\vfyh$ (or $\vfz$) into  $|\Sigma^L|$ i.i.d. channels, one for each of the components of the frequency vectors. Since $X^{ML}$ contains $M$ molecules, the input frequency vector is constrained by $\sum\limits_{i=1}^{|\Sigma^L|}\fx(a_i)=M$. Let $P_{\fyh\fz|\fx}=P_{\fyh \mid \fx}P_{\fz \mid \fx}$ denote the channel transition probability between the components $\fx(a_i)$ and $(\fyh(a_i), \fz(a_i))$, for $a_i \in \Sigma^L$. Since $\fyh$ (resp.\ $\fz$) is a sum of $\fx$ number of independent Ber($p_0$) (resp.\ Ber($q_0$)) random variables, the channel transition probabilities are 

\begin{align*}
    P_{\fyh \mid \fx}(j \mid i)= \begin{cases}
    \binom{i}{j}p_0^{i-j}(1-p_0)^j, & \text{if $j \leq i$ and }\\
            0, & \text{otherwise}
    \end{cases}
\end{align*}
and 
\begin{align*}
    P_{\fz \mid \fx}(j \mid i)= \begin{cases}
    \binom{i}{j}q_0^{i-j}(1-q_0)^j, & \text{if $j \leq i$ and }\\
            0, & \text{otherwise}
    \end{cases}
\end{align*}
where $\{0,1,\ldots,M\}$ is the alphabet for $\fx$, $\fyh$ and $\fz$. 

\begin{lemma}\label{lem:degraded_special}
If $q_0\geq p_0$, then the channel $P_{\fz \mid \fx}$ is a degraded version of $P_{\fyh \mid \fx}$, i.e., there exists a channel $Q$ such that
\begin{align*}
    P_{\fz \mid \fx}(j|i)=\sum \limits_{k=0}^M Q(j|k)P_{\fyh \mid \fx}(k|i)
\end{align*}
for all $i,j \in \{0,1,\ldots,M\}$.
\end{lemma}
\begin{proof}
The channel $Q$ defined by
\begin{align*}
    Q(j|i)= \begin{cases}
    \binom{i}{j}(\frac{q_0-p_0}{1-p_0})^{(i-j)}(\frac{1-q_0}{1-p_0})^j, & \text{if $j \leq i$ and }\\
            0, & \text{otherwise}
    \end{cases}
\end{align*}
satisfies the degradation condition. This can be easily verified using the identity $\sum_{k=j}^{i}\binom{i}{k} \binom{k}{j}(x-1)^{k-j}=\binom{i}{j}x^{i-j}$, which follows from the fact that $\binom{i}{k} \binom{k}{j}=\binom{i}{j}\binom{i-j}{k-j}$. 
\end{proof}

Using the above degradation relation ($W-\vfx-\vfyh-\vfz$), we can bound the rate as
\begin{align*}
   (1-\delta_M)R - \delta''_M - \rk
   \leq  \frac{1}{ML} [I(W; \vfyh) - I(W; \vfz)] \leq \frac{1}{ML} I(\vfx; \vfyh\mid \vfz) \leq \frac{1}{ML}H(\vfyh \mid \vfz).
\end{align*}
We will now bound $H(\vfyh \mid \vfz)$ using a counting argument. To this end, define the events
\begin{align*}
    E_Y := \left\lbrace \left| \frac{\lVert \vfyh\rVert_1 }{M} - (1-p_0)\right| \geq \delta \right\rbrace \quad \text{ and } \quad
    E_Z := \left\lbrace \left| \frac{\lVert \vfz\rVert_1 }{M} - (1-q_0)\right| \geq \delta \right\rbrace,
\end{align*}
where $\lVert \cdot \rVert_1 $ denotes the $L_1$ norm of a vector. Since $\lVert \vfyh\rVert_1$ (resp. $\lVert \vfz\rVert_1$) is a sum of $M$ independent Bernoulli random variables with parameter $p_0$ (resp. $q_0$), we have, by Hoeffding's inequality, 
$\mathbb{P}(E_Y) \leq 2 e^{-2M\delta^2}$ and $\mathbb{P}(E_Z) \leq 2 e^{-2M\delta^2}$. Therefore $\mathbb{P}(E_Y^c \cap E_Z^c) \geq 1-4e^{-2M\delta^2}$. Let $\mathbf{1}_E$ be the indicator function of the event $E:=E_Y^c \cap E_Z^c$. The following combinatorial lemma allows us to bound the number of frequency vectors $\vfyh$ possible given a realization of  $\vfz$.
\begin{lemma}
Let  $\mathbf{b}=(b_1, b_2,\ldots,b_{\ell})$ be a vector with integer entries. Let $S_{\ell}(\mathbf{b},n)$ denote the number of possible integer vectors $\mathbf{s} = (s_1,s_2,\ldots,s_{\ell})$ such that $s_i \ge b_i$ for all $i$, and $\sum_{i=1}^{\ell} s_i = n$.
Then,
\begin{align*}
    S_\ell(\mathbf{b},n) = \binom{n-b_1-\cdots-b_\ell+\ell-1}{\ell-1}.
\end{align*}
\end{lemma}
\begin{proof}
Note that any vector $\mathbf{s}$ counted by $S_{\ell}(\mathbf{b},n)$ can be uniquely expressed as the sum of $\mathbf{b}$ and a non-negative integer vector whose components sum to $n-\sum_i b_i$. Therefore, $S_{\ell}(\mathbf{b},n)$ is equal to the number of non-negative integer vectors with $\ell$ components that sum to $n-\sum_i b_i$. The latter number is precisely $\binom{n-b_1-\cdots-b_\ell+\ell-1}{\ell-1}$.
\end{proof}

For a realization of $\vfz$ corresponding to the event $E$, let $\mathbf{b}=\vfz$ and $S_\ell(\mathbf{b},n)$ denote the possible realizations of $\vfyh$ in the event $E$. For this choice,  $n \leq \left \lceil M(1-p_0 +\delta)\right \rceil $,  $\ell= |\Sigma|^L$, and  $\mathbf{b}=(b_1, b_2,\ldots,b_\ell)$ satisfies $\sum_i b_i \geq \left \lfloor M(1-q_0 - \delta)\right \rfloor$. So, we have
\begin{align*}
    S_\ell(\mathbf{b},n) 
    & \leq \binom{\left \lceil M(1-p_0 +\delta)\right \rceil- \left \lfloor M(1-q_0 -\delta)\right \rfloor+|\Sigma|^L}{|\Sigma|^L-1}
\end{align*}
which holds because $n-b_1-\cdots-b_\ell+|\Sigma|^L-1\leq \left \lceil M(1-p_0 +\delta)\right \rceil- \left \lfloor M(1-q_0 +\delta)\right \rfloor+|\Sigma|^L$. We can use the fact that $\binom{a}{a-b}=\binom{a}{b}<
\left(\frac{ea}{b}\right)^b$ to further bound $S_{\ell}(\mathbf{b},n)$ as 
\begin{align*}
    S_\ell(\mathbf{b},n)\leq  \left[e\left(1+\frac{|\Sigma|^L}{M(q_0-p_0+2\delta)-2}\right)\right]^{M(q_0-p_0+2\delta) +3}
\end{align*}
This implies that
\begin{align*}
&H(\vfyh \mid \vfz,  \mathbf{1}_E=1) \le [M(q_0-p_0+2\delta) +3]\log \left[e\left(1+\frac{|\Sigma|^L}{M(q_0-p_0+2\delta)-2}\right)\right]
\end{align*}
Note also that $H(\vfyh \mid \vfz,  \mathbf{1}_E=0)\leq \log \binom{M+|\Sigma|^L-1}{M}\leq M\log\left( \frac{e(M+|\Sigma|^L-1)}{M}\right)$. Thus

\begin{align}
   (1-&\delta_M)R - \delta''_M - \rk \leq \frac{1}{ML} H(\vfyh \mid \vfz) \notag \leq \frac{1}{ML}  H(\vfyh, \mathbf{1}_E \mid \vfz) 
    \leq \frac{1}{ML}\left[1+H(\vfyh \mid \vfz,  \mathbf{1}_E)\right] \notag \\
     &\leq \frac{1}{ML}\left[1+\mathbb{P}(E)H(\vfyh \mid \vfz,  \mathbf{1}_E=1)+\mathbb{P}(E^c)H(\vfyh \mid \vfz,  \mathbf{1}_E=0)\right] \notag \\
     &\leq  \frac{1}{ML}\left[1+\mathbb{P}(E)(M(q_0-p_0+2\delta) +3)\log \left(e\left(1+\frac{|\Sigma|^L}{M(q_0-p_0+2\delta)-2}\right)\right) \right.\notag\\
     &\mkern 350mu \left.+\mathbb{P}(E^c)M\log\left( \frac{e(M+|\Sigma|^L-1)}{M}\right)\right] \label{R_ineq}
\end{align}
Observe that
\begin{align*}
    \lim_{M\to\infty}\frac{1}{L}\log \left(e\left(1+\frac{|\Sigma|^L}{M(q_0-p_0+2\delta)-2}\right)\right)= \left(1-\frac{1}{\beta}\right).
\end{align*}
Also,
\begin{align*}
    \lim_{M\to\infty}\frac{1}{L}\log\left( \frac{e(M+|\Sigma|^L-1)}{M}\right)= \left(1-\frac{1}{\beta}\right), 
\end{align*}
which implies that \begin{align*}
    \liminf_{M\to\infty}\frac{\mathbb{P}(E^c)}{L}\log\left( \frac{e(M+|\Sigma|^L-1)}{M}\right)= 0.
\end{align*}
By taking limits on both sides of the inequality for $R$ in \eqref{R_ineq}, we get $R\leq (q_0-p_0+2\delta) \left(1-\frac{1}{\beta}\right)$.
Since $\delta$ is arbitrary, we conclude that  $ C_s - \rk\leq (q_0-p_0)\left(1-\frac{1}{\beta}\right).$

\section{Block-erasure wiretap channel}\label{app:block-erasure}
In this appendix, we will formally describe a block-erasure channel and a block-erasure wiretap channel. For these channels, we will also present the capacity results.  

\subsection{Block-Erasure Channel}
Let $m$ and $l$ be non-negative integers, and let $\mc{X}$ be a finite set.  Given a sequence $a^{ml}=(a_1, a_2, \ldots, a_{ml})$ and  $i \in [m]$, we refer to the contiguous subsequence $\underline{a}_i:=(a_{(i-1)l+1}, \ldots, a_{il})$ as the $i$th block (of length $l$) of $a^{ml}$. A \emph{block-erasure channel with parameter $\epsilon$}, denoted by B-EC($\epsilon, m, l$), takes a sequence $X^{ml}  \in \mc{X}^{ml}$ as the input and acts sequentially on blocks of length $l$ in the following manner. For each $i \in [m]$, the channel considers the $i$th block $\uX_i$ and either outputs $l$ erasure symbols $(\Delta, \ldots, \Delta) =:\underline{\Delta}$ with probability $\epsilon$ or noiselessly transmits $\uX_i$ with probability $1-\epsilon$, independently of the rest of the blocks. Define $\mc{Y}:=\mc{X} \cup \Delta$.  Let $Y^{ml}=(Y_1, Y_2, \ldots, Y_{ml}) \in \mc{Y}^{ml}$ be the output of B-EC($\epsilon, m, l$). The transition probability is given by 
\begin{align*}
    P_{Y^{ml}|X^{ml}}(y^{ml}|x^{ml})= \prod_{i=1}^{m}P_{\uY_i|\uX_i}(\uy_i|\ux_i),
\end{align*} 
where
\begin{align}\label{eq:cond_prob_bec}
    P_{\uY_i|\uX_i}(\uy_i|\ux_i) = \left\{\begin{array}{ll}
        \epsilon & \text{if } \uy_i=\underline{\Delta},\\
        1-\epsilon & \text{if } \uy_i=\ux_i,\\
        0 & \text{otherwise},\end{array} \right.
\end{align}
for $i \in [m]$, $\ux_i \in \mc{X}^l$, and $\uy_i\in \mc{Y}^{l}$.

\begin{figure}[h]
\centering
\resizebox{0.9\width}{!}{\tikzstyle{block}=[rectangle, draw, thick, minimum width=2em, minimum height=2em]

   



\begin{tikzpicture}[node distance=3.2cm,auto,>=latex']
    \node (x) {$(X_1, \ldots, X_{ml})$};
    \node [block, align=center] (main) [right of=x,] {B-EC($\epsilon, m, l$)};
   \node (y) [right of=main] {$(Y_1, \ldots, Y_{ml})$};

    \draw[->, thick] (x.east) --  (main.west);

    \draw[->, thick] (main) --  (y);

\end{tikzpicture}}
\caption{Block-erasure channel with erasure probability $\epsilon$ acting on blocks of length $l$.}
\label{fig:erasure_p2p}
 \end{figure}
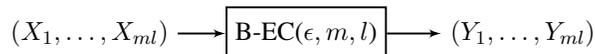

We are mainly interested in the channel capacity of a sequence of channels $\left\{\text{B-EC}(\epsilon, m, l)\right\}_{m=1}^{\infty}$, where $l\in \mathbb{N}$ is some function of $m$; specifically, in the DNA storage situation, $l \sim (\beta-1)\log m$ for some constant $\beta>1$. It is not difficult to show that the channel capacity of $\left\{\text{B-EC}(\epsilon, m, l)\right\}_{m=1}^{\infty}$ is $(1-\epsilon)\log|\mc{X}|$. 

\begin{remark}
Suppose $l$ is a constant. Then, if we view each block as a symbol from $\mc{X}^l$, the resulting channel becomes a discrete memoryless channel (with input alphabet $|\mc{X}|^l$), whose  channel capacity \cite[Theorem~7.7.1]{cover_thomas} is  $\frac{1}{l}(1-\epsilon)\log|\mc{X}|^l=(1-\epsilon)\log|\mc{X}|$. But, if $l$ grows with $m$, then the input alphabet $|\mc{X}|^l$ also grows with $m$. In this case, though the resulting channel is memoryless, the alphabet size is not fixed. Hence, we cannot directly use results for discrete memoryless channels. Nonetheless, capacity results for general ``information-stable'' channels from \cite{Dobrushin,vembu95,verdu_han_capacity} can be applied to this channel without treating it as a memoryless channel. 
\end{remark}

\begin{definition}[\cite{Dobrushin,vembu95,verdu_han_capacity}]
    A channel $\bW=\{\mathrm{W}^{n}\}_{n \geq 1}$ is said to be \emph{information-stable} if there exists an input process $\bX=\{X^n\}_{n \geq 1}$ such that 
    \begin{align}\label{cond:info_stable}
        \frac{i_{X^nY^n}\left(X^n;Y^n\right)}{n C_n(\mathrm{W}^{n})} \longrightarrow 1
    \end{align}
    in probability, where $i_{X^nY^n}\left(x^n;y^n\right):=\log \frac{P_{Y^n|X^n}\left(y^n|x^n\right)}{P_{Y^n}(y^n)}$ and $C_n(\mathrm{W}^{n}):= \sup_{X^n} \frac{1}nI\left(X^n;Y^n\right)$.
\end{definition}

\begin{theorem}[\cite{Dobrushin,vembu95,verdu_han_capacity}]\label{thm:capacity_stable}
    For an information-stable channel $\bW$, the channel capacity\footnote{ 
    We can even apply this result for arbitrary sequences \cite[p. 1148]{verdu_han_capacity} of channels so long as we use a scaling in \eqref{cond:info_stable} that corresponds to the sequence.} $C(\bW)$ of $\bW$ is given by
    \begin{align*}
        C(\bW)=\liminf_{n \to \infty} \sup_{X^n} \frac{1}nI\left(X^n;Y^n\right).
    \end{align*}
\end{theorem}

\begin{corollary}\label{cor:capacity_bec}
 The sequence of block-erasure channels \emph{$\bW=\left\{\text{B-EC}(\epsilon, m, l)\right\}_{m=1}^{\infty}$} is information-stable, and the channel capacity is $C(\bW)=(1-\epsilon)\log|\mc{X}|$.
\end{corollary}
\begin{proof}
First note that $C_{ml}(W^{ml})=C_{ml}(\text{B-EC}(\epsilon, m, l))=\sup_{X^{ml}} \frac{1}{ml}I\left(X^{ml};Y^{ml}\right)=(1-\epsilon)\log|\mc{X}|$, which follows from 
\begin{align}
    I\left(X^{ml};Y^{ml}\right)= H(Y^{ml})-H(Y^{ml}|X^{ml})& \leq \sum_{i=1}^{m}H(\uY_i)-\sum_{i=1}^{m}H(\uY_i|X^{ml},\uY_1,\ldots,\uY_{i-1})\notag\\
    & \stackrel{(a)}{=} \sum_{i=1}^{m}H(\uY_i)-\sum_{i=1}^{m}H(\uY_i|\uX_i)\notag\\
    & = \sum_{i=1}^{m}I(\uX_i;\uY_i)\notag\\
    & \leq m\sup_{\uX} I(\uX;\uY)
    \stackrel{(b)}{=} ml(1-\epsilon)\log|\mc{X}|,\label{eq:sup_p2p_capacity_bec}
\end{align}
where $(a)$ is due to the fact that given an input block, the corresponding output block does not depend on the rest of the blocks, and $(b)$ follows from the capacity expression of an erasure channel with input alphabet $\mc{X}^l$. The above supremum is achieved by the input $\hat{X}^{ml}$ whose components are  independent and uniformly distributed random variables. For this input process $\bX=\{\hat{X}^{ml}\}_{m \geq 1}$, we will now show \eqref{cond:info_stable}. Let $\bY=\{\hat{Y}^{ml}\}_{m \geq 1}$ be the corresponding output of $\bW$. Observe that for $i \in [m]$,
\begin{align*}
    P_{\uhY_i(\uy_i)} = \left\{\begin{array}{ll}
        \epsilon & \text{if } \uy_i=\underline{\Delta},\\
        \frac{1-\epsilon}{|\mc{X}|^l} & \text{if } \uy_i \in \mc{X}^l,\\
        0 & \text{otherwise},\end{array} \right.
\end{align*}
which together with \eqref{eq:cond_prob_bec} yield
\begin{align*}
    i_{\hat{X}^{ml}\hat{Y}^{ml}}\left(\hat{X}^{ml};\hat{Y}^{ml}\right)
    &=\sum_{i=1}^{m}\log \frac{P_{\uhY_i|\uhX_i}\left(\uhY_i|\uhX_i\right)}{P_{\uhY_i}(\uhY_i)}= l\log|\mc{X}| \sum_{i=1}^{m}\mathbf{1}_{\left\lbrace\uhY_i = \uhX_i\right\rbrace},
\end{align*}
where the last equality holds with probability one.
By the weak law of large numbers, for every $\gamma>0$, 
\begin{align}\label{eq:convergence}
\lim_{m \to \infty}& \Pr\left\{\left|\frac{i_{\hat{X}^{ml}\hat{Y}^{ml}}\left(\hat{X}^{ml};\hat{Y}^{ml}\right)}{{ml} C_{ml}(W^{ml})} -  1\right| > \gamma\right\}  = \lim_{m \to \infty} \Pr\left\{\left|\frac{\sum_{i=1}^{m}\mathbf{1}_{\left\lbrace\uhY_i = \uhX_i\right\rbrace}}{m(1-\epsilon)}-1 \right|> \gamma \right\} = 0,
\end{align}
which proves that $\bW$ satisfies the information-stable condition \eqref{cond:info_stable}. So, by Theorem~\ref{thm:capacity_stable}, the channel capacity is  $C(\bW)=\liminf_{m \to \infty} C_{ml}(\text{B-EC}(\epsilon, m, l)) = (1-\epsilon)\log|\mc{X}|$.
\end{proof}

\begin{remark}
    The convergence in \eqref{eq:convergence} is  exponential in $m$, i.e., for every $\alpha>0$,
    \begin{align}\label{eq:exp_convergence}
    \Pr\left\{\left|\frac{\sum_{i=1}^{m}\mathbf{1}_{\left\lbrace\uhY_i = \uhX_i\right\rbrace}}{m}-(1-\epsilon) \right|> \alpha \right\} \leq 2\cdot 2^{-2m\alpha^2},
    \end{align}
    which follows from  application of Hoeffding's inequality to the indicator random variables that are independent. Therefore, the sequence of block-erasure channels $\left\{\text{B-EC}(\epsilon, m, l)\right\}_{m=1}^{\infty}$ are indeed ``exponentially\footnote{Strictly speaking, the notion of exponentially information-stable requires the left-hand side of \eqref{eq:exp_convergence} to decay exponentially fast in $ml$; however, exponential decay in $m$ is sufficient for our purposes in the next section.} information-stable" \cite[Remark~3]{bloch_strong_secrecy}.
\end{remark}

\subsection{Block-Erasure Wiretap Channel}
In a \emph{block-erasure wiretap channel}, both the main and wiretap channels are block-erasure channels. We use the notation $\text{B-EWTC}(\epm, \epw,  m, l)$ to denote a block-erasure wiretap channel with $\text{B-EC}(\epm, m, l)$ and $\text{B-EC}(\epw, m, l)$ as the main and wiretap channels, respectively, where $m, l \in \mathbb{N}$ --- see Fig.~\ref{fig:erasure_wtc}.

\begin{figure}[h]
\centering
\resizebox{0.9\width}{!}{\tikzstyle{block}=[rectangle, draw, thick, minimum width=2em, minimum height=2em]

\begin{tikzpicture}[node distance=3.2cm,auto,>=latex']
    \node (x) {$(X_1, \ldots, X_{ml})$};
    \node (dummy) [right of=x] {};
    \node [block, align=center] (main) [above of = dummy, node distance=0.8cm] {$\text{B-EC}(\epm, m, l)$};
    \node [block, align=center] (wiretap) [below of = dummy, node distance=0.8cm] {$\text{B-EC}(\epw, m, l)$};
    \node (y) [right of=main] {$(Y_1, \ldots, Y_{ml})$};
    \node (z) [right of=wiretap] {$(Z_1, \ldots, Z_{ml})$};

    \draw[->, thick] (x.east) --  (main.west);
    \draw[->, thick] (x.east) --  (wiretap.west);
    \draw[->, thick] (main) --  (y);
    \draw[->, thick] (wiretap) --  (z);

\end{tikzpicture}

    

}
\caption{Block-erasure wiretap channel $\text{B-EWTC}(\epm, \epw,  m, l)$.}
\label{fig:erasure_wtc}
\end{figure}

Let $\bW=\left\{\text{B-EWTC}(\epm, \epw,  m, l)\right\}_{m=1}^{\infty}$ and $l\sim (\beta -1)\log m$. The \emph{secrecy capacity} of a block-erasure wiretap channel is defined similarly to the secure storage capacity of a DNA wiretap channel in Section~\ref{sec:problem} and Section~\ref{sec:problem_shared_key}. To be precise, let $W$ be a uniform random variable taking values in $\mathcal{W}\triangleq \left[2^{\lfloor mlR \rfloor}\right]$ for some $R \geq 0$. A uniform random variable $K \in \mathcal{K}\triangleq \big[2^{\lfloor ml\trk \rfloor}\big]$, called a \emph{shared key}, is available only at Alice and Bob for some $\trk \geq 0$. Alice encodes her message $W$, which is independent of $K$, into a length-$ml$ sequence over the alphabet $\mc{X}$ using an encoding function $f: \mathcal{W} \times \mc{K} \longrightarrow \mc{X}^{ml}$. The encoded sequence $X^{ml}$ is the input to $\text{B-EWTC}(\epm, \epw,  m, l)$. Bob observes the output $Y^{ml}$ of the main channel $\text{B-EC}(\epm, m, l)$, and similarly, Eve observes the output $Z^{ml}$ of the wiretap channel $\text{B-EC}(\epw, m, l)$. Bob uses a decoding function $g: \mc{Y}^{ml} \times \mc{K} \longrightarrow \mathcal{W} \cup \{\textbf{e}\}$, where $\textbf{e}$ denotes an error. We say that a secure message transmission rate $R$ is achievable  with a shared key of rate $\trk$ if there exists a sequence of pairs of encoding and decoding functions $\{(f,g)\}_{m=1}^{\infty}$ satisfying Bob's recoverability condition,
 \begin{align} \label{eq:bob_recoverability_erasurewtc}
     \Pr\left\{g(Y^{ml}, K) \neq W\right\} \to 0,
 \end{align}
 and Eve's (strong) secrecy condition,
 \begin{align}\label{eq:eve_secrecy_erasurewtc}
    I(W ; Z^{ml}) \to 0,
 \end{align}
 as $m \to \infty$.  The secrecy capacity of a block-erasure wiretap channel with a shared key of rate $\trk$ is defined by 
  $C_s(\bW, \trk) \triangleq \sup\{R: R \text{ is achievable with a shared key of rate } \trk \}$.

As the block-erasure wiretap channel is not memoryless, if $\trk=0$, then we can obtain a characterization of $C_s(\bW, 0)$ using Theorem~2 of \cite{bloch_strong_secrecy} together with Remark~3 of \cite{bloch_strong_secrecy}, whose results can be applied to arbitrary wiretap channels. These results were obtained using ``channel resolvability" arguments. However, in the case of $\trk>0$, the secrecy capacity is known for memoryless wiretap channels from the work of \cite{rafael_shared_key, kang_shared_key}. We combine the ideas from \cite{rafael_shared_key, kang_shared_key} with the channel resolvability arguments of \cite{bloch_strong_secrecy} to derive an expression for the (strong) secrecy capacity $C_s(\bW, \trk)$ in the following theorem. To simplify the presentation of the result, we assume $\mc{X}=\{0,1\}$; however, for an arbitrary finite alphabet $\mc{X}$, there will be an extra factor, $\log |\mc{X}|$, accompanying the terms $\epw-\epm$ and $1-\epm$ in the result.

\begin{theorem}\label{thm:capacity_bewtc}
For a sequence of block-erasure wiretap channels \emph{$\bW=\left\{\text{B-EWTC}(\epm, \epw,  m, l)\right\}_{m=1}^{\infty}$} with a shared key available to Alice and Bob of rate $\trk$ and $l \sim (\beta -1)\log m$, we have 
\emph{
\begin{align}\label{eq:capacity_bewtc}
    C_s(\bW, \trk) = \min\left\{\trk+[\epw-\epm]^{+}, 1-\epm\right\}.
\end{align} }
Here $[u]^{+} \triangleq \max\{u,0\}$.
\end{theorem}
\begin{proof}
We will prove \eqref{eq:capacity_bewtc} by considering two separate cases: $\epw <\epm$ and $\epw \geq \epm$.\\

 $(i) \ \epw < \epm:$  In this case, we need to show that \begin{align}\label{eq:capacity_bewtc_w_leq_m}
    C_s(\bW, \trk) = \min\left\lbrace\trk, 1-\epm \right\rbrace.
\end{align}
\underline{Achievability part}: If $\trk > 1-\epm$, then Alice and Bob can utilize only a part of the key $K$ of rate $1-\epm$ for message transmission and discard the remaining part. So, if  we show that when $\trk \leq 1-\epm$, $\trk$ is an achievable secure message transmission rate, then $\min\left\lbrace\trk, 1-\epm \right\rbrace \leq C_s(\bW, \trk)$.

Let $\trk \leq 1-\epm$. Alice uses the ``one-time pad" scheme to encode the message $W \in \big[2^{\lfloor ml\trk \rfloor}\big]$ using the shared key $K \in \big[2^{\lfloor ml\trk \rfloor}\big]$. It means that Alice computes ciphertext $$C:= W + K \bmod 2^{\lfloor ml\trk \rfloor}$$ and encodes $C$ into $\mc{X}^{ml}$ for transmission on $\text{B-EWTC}(\epm, \epw,  m, l)$. It is easy to see that $C$ is uniform over its range because $W$ and $K$ are independent uniform random variables.

Since the rate of the ciphertext $\trk$ is less than the channel capacity of the main block-erasure channel $1 - \epm$, Corollary~\ref{cor:capacity_bec} guarantees the existence of a sequence of codes such that $\lim_{m \to \infty}\Pr\{\widehat{C} \neq C\}=0,$ where $\widehat{C}$ is Bob's estimate for $C$. As Bob also has the shared key $K$, he can recover the message $W$ from $\widehat{C}$ by computing $\widehat{W}:=\widehat{C} + K$. Therefore, we have
\begin{align*}
    \Pr\lbrace\widehat{W} \neq W\rbrace = \Pr\lbrace\widehat{C} + K \neq C + K\rbrace= \Pr\{\widehat{C} \neq C\}\to 0
\end{align*}
as $m \to \infty$, which satisfies  Bob's recoverability condition \eqref{eq:bob_recoverability_erasurewtc}.

It remains to verify that the above scheme achieves the strong secrecy condition \eqref{eq:eve_secrecy_erasurewtc}. By virtue of the underlying one-time pad, we can, in fact, argue the strongest form of secrecy, namely, the perfect secrecy, i.e., $I(W; Z^{ml})=0$.  To see this, first note the above encoding scheme induces a distribution on the random variables that satisfy the Markov chain $(W,K) - C - X^{ml} - \left(Y^{ml}, Z^{ml}\right)$. By data processing inequality, $I(W; Z^{ml}) \leq I(W; C)=0$, where the last inequality follows from the independence of $W$ and $C$. Hence,  $\trk$ is an achievable secure message transmission rate.

\noindent\underline{Converse part:} Suppose that there is a sequence of pairs of encoding and decoding functions satisfying conditions \eqref{eq:bob_recoverability_erasurewtc} and \eqref{eq:eve_secrecy_erasurewtc} with secure message transmission rate $R$. First let us show that $R \leq 1- \epm$ irrespective of how $\epw$ and $\epm$ are related. 
\begin{align}
        mlR -1  \leq \log |\mathcal{W}|  = H(W) = H(W|K)
        & \stackrel{(a)}{=} H(W|K)-H(W \mid Y^{ml}, K)+mlR\delta_m +1\notag\\
        & = I(W ; Y^{ml} \mid K)+ mlR\delta_m +1\notag\\
        & \leq  I(W, K ; Y^{ml})+ mlR\delta_m +1 \label{ineq:conv_12}\\
        & \stackrel{(b)}{\leq}  I(X^{ml} ; Y^{ml})+ mlR\delta_m +1\notag
\end{align}
for $\delta_m \to 0$ as $m \to \infty$, where $(a)$ is because of  Bob's recoverability condition \eqref{eq:bob_recoverability_erasurewtc} and Fano's inequality, and  $(b)$  follows by applying data processing inequality to the Markov chain $(W, K)-X^{ml}-Y^{ml}.$ Therefore, we have
\begin{align*}
    R\leq \liminf_{m \to \infty}\frac{1}{ml} \sup_{X^{ml}} I\left(X^{ml};Y^{ml}\right)=1-\epm,
\end{align*}
where the equality follows from  \eqref{eq:sup_p2p_capacity_bec} of Corollary~\ref{cor:capacity_bec}. This implies that 
\begin{align}\label{ineq:conv_ew leq em_1}
    C_s(\bW, \trk) \leq  1-\epm
\end{align}
which holds for all $\epm$ and $\epw$.

Now we show that if $\epw < \epm$, then $C_s(\bW, \trk)$ is upper bounded by $\trk$. We proceed from the inequality \eqref{ineq:conv_12} as follows:
\begin{align}
        mlR  \leq  I(W, K ; Y^{ml})+ mlR\delta_m +2&  = I(K ; Y^{ml}|W)+ I(W ; Y^{ml})+ mlR\delta_m +2\notag \\
        & \leq H(K) + I(W ; Y^{ml})+ mlR\delta_m +2 \notag \\
        & \stackrel{(b)}{\leq} H(K) + I(W ; Y^{ml}) - I(W ; Z^{ml})+ mlR\delta_m +\delta'_m+2, \label{ineq:conv_13}
\end{align}
where $(b)$ follows from Eve's secrecy condition \eqref{eq:eve_secrecy_erasurewtc} with $\delta'_m \to 0$ as $m \to \infty$. Since the difference $I(W ; Y^{ml}) - I(W ; Z^{ml})$ depends only on the marginal distributions $P_{Y^{ml}|W}$ and $P_{Z^{ml}|W}$, we can evaluate this difference for a new distribution $\tilde{P}_{Y^{ml}, Z^{ml}|W}$ with the same marginals. If $\epw<\epm$, then the output $Y^{ml}$ of $\text{B-EC}(\epm, m, l)$ is a degraded version of the output $Z^{ml}$ of $\text{B-EC}(\epw, m, l)$. This allows us to consider a new distribution $\tilde{P}_{Y^{ml}, Z^{ml}|W}$ that satisfies the Markov chain $W-X^{ml}-Z^{ml}-Y^{ml}$. Hence, we can further bound \eqref{ineq:conv_13} as
\begin{align*}
        mlR &\leq H(K) + I(W ; Y^{ml}) - I(W ; Z^{ml})+ mlR\delta_m +\delta'_m+2 \\
        & \stackrel{(a)}{=} H(K) - I(W ; Z^{ml}| Y^{ml})+ mlR\delta_m +\delta'_m+2 \\
        & \leq ml\trk + mlR\delta_m +\delta'_m+2,
\end{align*}
where $(a)$ follows from the Markov chain $W-X^{ml}-Z^{ml}-Y^{ml}$ for the coupled distribution $\tilde{P}_{Y^{ml}, Z^{ml}|W}$. By normalizing and taking limits on both sides of the above inequality, we get $R \leq \trk$, which means that 
\begin{align}\label{ineq:conv_ew leq em_2}
     C_s(\bW, \trk) \leq  \trk.
\end{align}
By combining \eqref{ineq:conv_ew leq em_1} and \eqref{ineq:conv_ew leq em_2}, we prove the converse result $C_s(\bW, \trk) \leq \min\left\lbrace\trk, 1-\epm \right\rbrace.$
\\

\noindent$(ii) \ \epw \geq \epm:$ In this case, \eqref{eq:capacity_bewtc} becomes
\begin{align*} 
    C_s(\bW, \trk) = \min\left\{\trk+(\epw-\epm), 1-\epm\right\}.
\end{align*}
\noindent\underline{Achievability part:} We will show that $\min\{\trk+(\epw-\epm), 1-\epm\}$ is an  achievable secure message transmission rate. Fix 
$\gamma > 0$. Let $K \in  \big[2^{\lfloor ML\trk \rfloor}\big]$ be the shared key of rate $\trk$ available with Alice and Bob. Let  $R\geq 0$ and $\tR\geq 0$ be the rates of the message $W \in \left[2^{\lfloor MLR \rfloor}\right]$ and the auxiliary message $\tW \in  \big[2^{\lfloor ML\tR \rfloor}\big]$, respectively. The main purpose of using an auxiliary message is to confuse Eve. We will now present a scheme that achieves the desired rate.

\begin{itemize}[leftmargin=*]
    \item \emph{Codebook generation:} We denote by $P_{\hat{X}^{ml}}$ the distribution of the sequence  $\hat{X}^{ml}$ whose components are independent and uniformly distributed random variables. For each $(w,k,\tw) \in \left[2^{\lfloor mlR \rfloor}\right] \times \big[2^{\lfloor ml\trk \rfloor}\big] \times \big[2^{\lfloor ml\tR \rfloor}\big]$, pick a codeword $x^{ml}_{w,k,\tw} \in \mc{X}^{ml}$  independently of the other codewords according to the distribution $P_{\hat{X}^{ml}}$. The collection of these codewords is called a \emph{codebook}. Let $\code$ denote a realization of the random codebook $\randcode$.
    
    \item \emph{Encoding:}  Suppose that Alice and Bob observe the shared key value $K=k$, and Alice wishes to transmit message $W=w$ to Bob. Then, given a codebook $\code$, Alice chooses an auxiliary message $\Tilde{W}=\tw$  uniformly at random from  $\big[2^{\lfloor ml\tR \rfloor}\big]$ and transmits the codeword  $x^{ml}_{w,k,\tw} \in \code$ corresponding to the triple $(w,k,\tw)$.
     \item \emph{Decoding:} After receiving $y^{ml}$ (the output of the main channel), Bob recovers the transmitted message using the shared key value $K=k$ as follows. Let 
     \begin{align*}
         \mc{T}^{ml}\triangleq\left\lbrace\left(x^{ml}, y^{ml}\right) \in \mc{X}^{ml} \times \mc{Y}^{ml} :  \left|\frac{1}{ml}\log \frac{P_{\hat{Y}^{ml}|\hat{X}^{ml}}\left(y^{ml}|x^{ml}\right)}{P_{\hat{Y}^{ml}}\left(y^{ml}\right)} - (1-\epm)\right|< \gamma\right\rbrace,
     \end{align*}
    where $P_{\hY^{ml}|\hX^{ml}}$ is the transitional probability of the main block-erasure channel, and $P_{\hY^{ml}}$ is the output distribution of  Bob's channel  when the input distribution is $P_{\hX^{ml}}$. If there exists a unique index pair $(w,\tw)$ such that   $x^{ml}_{w,k,\tw} \in \code$ and $\left(x^{ml}_{w,k,\tw}, y^{ml}\right)\in \mc{T}^{ml}$, then Bob outputs $\widehat{W}=w$. In any other case, Bob declares an error. 
\end{itemize}

Let $P_e(\code)$ denote the average probability of error in recovering the transmitted message when codebook $\code$ is used. We write $\Bar{P}_e$ to denote $P_e(\code)$ averaged over the codebook ensemble, i.e., $\bar{P}_e:=\E_{\randcode}\left[P_e(\randcode)\right]$. Let $\skl(\code):= D_{\op{\scriptscriptstyle KL}}\left(P_{W \Bar{Z}^{ml}}\|P_{W}P_{\Bar{Z}^{ml}}\right)=I(W;\Bar{Z}^{ml})$ and $\stv(\code):= \linebreak\dtv\left(P_{W\Bar{Z}^{ml}}, P_{W}P_{\Bar{Z}^{ml}}\right)$, where $P_{W, \Bar{Z}^{ml}}$ is the joint distribution induced by $\code$ on $W$ and the output of Eve's channel. Here, $D_{\op{\scriptscriptstyle KL}}$ and $\dtv$ denote the Kullback-Liebler divergence and the total variation distance between distributions, respectively. 

\begin{lemma}\label{lem:code_exist_bewtc}
Let $\gamma>0$. There exists a sequence of codes $\{\code\}_{m=1}^{\infty}$ such that 
\begin{enumerate}
    \item if \emph{$R+\tR \leq (1-{\epm})-2\gamma$}, then $\lim_{m \to \infty} P_e(\code) =0,$
    \item if \emph{$\trk+\tR \geq (1-{\epw})+\gamma$}, then $\lim_{m \to \infty} \skl(\code) =0.$
\end{enumerate}
\end{lemma}

\noindent\emph{Proof of Lemma~\ref{lem:code_exist_bewtc}.} We prove Lemma~\ref{lem:code_exist_bewtc} along the lines of the proofs of Lemma~4 and Lemma~5 of \cite{bloch_strong_secrecy} by appropriately incorporating a shared key. We crucially use the condition $l \sim (\beta-1)\log m$ and the fact that Eve's  channel $\text{B-EC}(\epw, m, l)$ is exponentially information-stable \eqref{eq:exp_convergence} to prove this lemma.   

For $\gamma>0$, we will show that if $R+\tR \leq (1-\epm)-2\gamma$ and  $\trk+\tR \geq (1-\epw)+\gamma$, then 
\begin{align}\label{ineq:bounds_0}
    \bar{P}_e < \delta_m \quad \text{\  and  \ } \quad E\left[\stv\left(\randcode\right)\right] <  2^{-\alpha m},
\end{align} 
for some $\delta_m \to 0$ as $m \to \infty$ and $\alpha >0$. Then, by Markov's inequality,
we can conclude that there exists a sequence of  codebooks satisfying  
\begin{align}\label{ineq:bounds_1}
    P_e(\code) < 2\delta_m \quad \text{\  and  \ } \quad \stv\left(\code\right) <  2^{-\frac{\alpha}{2} m}.
\end{align} 
for some $\delta_m \to 0$ as $m \to \infty$ and $\alpha >0$. Now, by applying the following lemma, we can find the limit of $\skl(\code)$ for this sequence of codebooks.
\begin{lemma}[Lemma~1 of \cite{csiszar04}]\label{lemm:csis:secrecy} If $|\mc{W}|\geq 4$, then
    \begin{align*}
        \skl \leq \stv \log \frac{|\mc{W}|}{\stv}.
    \end{align*}
\end{lemma}
As $\stv\left(\code\right) <  2^{-\frac{\alpha}{2} m}$ for all sufficiently large $m$ and $|\mc{W}|= 2^{\lfloor mlR \rfloor}$, we have $\lim_{m \to \infty}\skl\left(\code\right) \leq \lim_{m \to \infty} [ 2^{-\frac{\alpha}{2} m} mlR - \stv\left(\code\right) \log \stv\left(\code\right)] =0,$
where the last equality is due to the fact that $l \sim (\beta-1)\log m$ and $\lim_{x \to 0} x\log x =0$. Lemma~\ref{lemm:csis:secrecy} together with \eqref{ineq:bounds_1} would complete the proof of Lemma~\ref{lem:code_exist_bewtc}.

Now it remains to prove \eqref{ineq:bounds_0}. Let us analyze the average probability of error for the above coding scheme. Given a codebook $\code$, denote by $D_{w,k,\tw}$ the set   $\big\lbrace y^{ml} : \big(x^{ml}_{w,k,\tw}, y^{ml}\big) \in \mc{T}^{ml} \big\rbrace$, and by $P_e(\code|W=w, K=k,\tW=\tw)$ the probability of error when Alice transmits the codeword word  $x^{ml}_{w,k,\tw} \in \code$ corresponding to the triple $(w,k,\tw)$. We have
    \begin{align*}
        &P_e(\code|W=w, K=k,\tW=\tw)\\  
        &\qquad\leq  \Pr \left\lbrace  D^c_{w,k,\tw}\Big|X_{w,k,\tw}^{ml}= x^{ml}_{w,k,\tw} \right\rbrace + \sum_{(r,t) \neq (w,\tw)} \Pr \left\lbrace   D_{r,k,t}\Big|X_{w,k,\tw}^{ml}= x^{ml}_{w,k,\tw}\right\rbrace,
    \end{align*}
    where the last inequality follows from the union bound. By taking expectation with respect to $W, K, \tW$, and codebook ensemble, we obtain
        \begin{align*}
        \bar{P}_e 
        = \E_{W,K,\tW} \E_{\randcode} P_e(\randcode|W,K,\tW) 
        & \stackrel{(a)}{=} \E_{\randcode} P_e(\randcode|W=1,K=1,\tW=1)\\
        &\leq   \E_{\randcode} \Pr \left\lbrace  D^c_{1,1,1}\big| X_{1,1,1}^{ml} \right\rbrace + \sum_{(w,\tw) \neq (1,1) } \E_{\randcode} \Pr \left\lbrace \left.   D_{w,1,\tw}\right| X_{1,1,1}^{ml}\right\rbrace\\
        &=   P_{\hX^{ml}\times\hY^{ml}} \left[\left(\mc{T}^{ml}\right)^c\right] + 2^{\lfloor mlR \rfloor  + \lfloor ml\tR \rfloor} \left(P_{\hX^{ml}}\times P_{\hY^{ml}}\right)\left(\mc{T}^{ml} \right)\\
        & < o(1) + 2^{ -ml[-R-\tR +(1-\epm)-\gamma]+2}:= \delta_m,
    \end{align*}
    where $(a)$ is due to the symmetry in random codebook generation.
    If we choose $R+\tR \leq (1-\epm)-2\gamma$, then $\bar{P}_e < \delta_m$, where $\delta_m \to 0$ as $m \to \infty$.

To establish the security of the above scheme, we will show that $\E\left[ \stv\left(\code\right) \right] \leq 2^{-\alpha m}$ for some $\alpha >0$. For a given codebook $\code$, we can write $\stv\left(\code\right) =  \dtv\left( P_{W}P_{\Bar{Z}^{ml}}, P_{W\Bar{Z}^{ml}}\right) = \E_{W}\left[ \dtv\left( P_{\Bar{Z}^{ml}}, P_{\Bar{Z}^{ml}|W}\right)\right] \linebreak \leq 2 \E_{W}\left[ \dtv\left( P_{\hZ^{ml}}, P_{\Bar{Z}^{ml}|W} \right)\right],$
where $P_{\hZ^{ml}}$ is the distribution of the output of Eve's channel when the input distribution is $P_{\hX^{ml}}$, and the inequality involves the use of the triangle inequality.
This implies that 
\begin{align*}
    \E \left[\stv\left(\randcode\right)\right] &\leq  2 \E_{\randcode} \E_{W}\left[ \dtv\left( P_{\hZ^{ml}}, P_{\Bar{Z}^{ml}|W} \right)\right]
     \stackrel{(a)}{=} 2 \E_{\randcode} \left[ \dtv\left( P_{\hZ^{ml}}, P_{\Bar{Z}^{ml}|W=1} \right)\right],
\end{align*}
where $(a)$ is, again, due to the symmetry in random codebook generation. If $\trk+\tR \geq (1-\epw)+\gamma$, then we can bound $\E_{\randcode} \left[ \dtv\left( P_{\hZ^{ml}}, P_{\Bar{Z}^{ml}|W=1} \right)\right]$ in following manner by using the ``channel-resolvability" arguments of \cite[Ch. 6]{han_book} and \cite{bloch_strong_secrecy}: For arbitrary $\tau>0$,
\begin{align*}
    \E_{\randcode} \left[ \dtv\left( P_{\hZ^{ml}}, P_{\Bar{Z}^{ml}|W=1} \right)\right] \leq & \ 2  \Pr\left\lbrace \frac{1}{ml}\log \frac{P_{\hZ^{ml}|\hX^{ml}}\big(\hZ^{ml}|\hX^{ml}\big)}{P_{\hZ^{ml}}\big(\hZ^{ml}\big)} > (1-\epw)  +\gamma + \frac{\log \rho}{ml}\right\rbrace \\
    &+ 2\Pr\left\lbrace \frac{1}{ml}\log \frac{P_{\hZ^{ml}|\hX^{ml}}\big(\hZ^{ml}|\hX^{ml}\big)}{P_{\hZ^{ml}}\big(\hZ^{ml}\big)} > (1-\epw) +\gamma \right\rbrace+ 2 \cdot  \frac{2^{-ml\frac{\gamma}{2}}}{\rho^2}\\
    &+ \frac{2}{\rho^2}\Pr\left\lbrace \frac{1}{ml}\log \frac{P_{\hZ^{ml}|\hX^{ml}}\big(\hZ^{ml}|\hX^{ml}\big)}{P_{\hZ^{ml}}\big(\hZ^{ml}\big)} > (1-\epw) +\frac{\gamma}{2}    \right\rbrace + 2\tau,
\end{align*}
where $\rho = \frac{2^{\tau}-1}{2}$. By choosing $\tau = 2^{-\eta m}$ with $0< \eta < \frac{\gamma^2}{4}$, we get $\rho = \frac{\ln 2}{2}2^{-\eta m} + o(2^{-\eta m})$, which satisfies $\frac{1}{\rho^2} \leq   2^{2\eta m+4}$ and $\lim_{m \to \infty}\frac{\log \rho}{ml} = 0$. Hence, for all large enough $m$, $\frac{\log \rho}{ml} > -\frac{\gamma}{2}$. Since Eve's  channel $\text{B-EC}(\epw, m, l)$ is exponentially information-stable \eqref{eq:exp_convergence}, we obtain the bound
\begin{align*}
    \E \left[\stv\left(\randcode\right)\right] &\leq   4 \cdot 2^{-\frac{\gamma^2}{2}m} + 4 \cdot 2^{-2\gamma^2m} + 4 \cdot 2^{2\eta m+4} \cdot 2^{-ml\frac{\gamma}{2}} + 4 \cdot 2^{2\eta m+4} \cdot 2^{-\frac{\gamma^2}{2}m} + 4\cdot 2^{-\eta m} \\
    & < 2^{-\alpha m}
\end{align*}
for some $\alpha>0$ and for all sufficiently large $m$, proving  \eqref{ineq:bounds_0}.

Since $\gamma$ is arbitrary in Lemma~\ref{lem:code_exist_bewtc}, every rate-triple  $(R, \trk, \tR)$ satisfying  $R+\tR \leq (1-\epm)$ and  $\trk+\tR \geq (1-\epw)$ is achievable. Observe that if $(\epw-\epm)+\trk \geq 0$, which trivially holds as $\epw \geq \epm$, and $(1-\epw)\geq \trk$, then the rate-triple $\left((\epw-\epm)+\trk, \trk, (1-\epw)-\trk\right)$ is achievable. On the other hand, if $\trk> (1-\epw)$, then by setting $\tR=0$, we can see that $((1-\epm), 0, \trk)$ is achievable. Therefore, 
\begin{align*}
    C_s(\bW, \trk) \geq \min\left\{\trk+(\epw-\epm), 1-\epm\right\}.
\end{align*}

\noindent\underline{Converse part:} We start from the inequality \eqref{ineq:conv_13}, and consider a new distribution $\tilde{P}_{Y^{ml}, Z^{ml}|W}$ with the marginals $P_{Y^{ml}|W}$ and $P_{Z^{ml}|W}$, as in the case of $\epw<\epm$. However, if $\epw\geq \epm$, then the output $Z^{ml}$ of $\text{B-EC}(\epw, m, l)$ is a degraded version of the output $Y^{ml}$ of $\text{B-EC}(\epm, m, l)$, which allows us to consider a new distribution $\tilde{P}_{Y^{ml}, Z^{ml}|W}$ that satisfies the Markov chain $W-X^{ml}-Y^{ml}-Z^{ml}$. For this coupled distribution, we can further bound \eqref{ineq:conv_13} as
\begin{align}
        mlR & \leq H(K) + I(W ; Y^{ml}) - I(W ; Z^{ml})+ mlR\delta_m +\delta'_m+2 \notag\\
        & \stackrel{(a)}{=} H(K) + I(W ; Y^{ml}| Z^{ml})+ mlR\delta_m +\delta'_m+2 \notag\\
        & \stackrel{(b)}{\leq} H(K) + I(X^{ml} ; Y^{ml}| Z^{ml})+ mlR\delta_m +\delta'_m+2 \notag\\
        & \leq H(K) + \sup_{X^{ml}}I(X^{ml} ; 
        Y^{ml}| Z^{ml})+ mlR\delta_m +\delta'_m+2 \notag\\
        & \stackrel{(c)}{=} H(K) + ml(\epw-\epm)+ mlR\delta_m +\delta'_m+2\label{ineq:conv_14}
\end{align}
where $(a)$ and $(b)$ are due to the Markov chain $W-X^{ml}-Y^{ml}-Z^{ml}$ for the coupled distribution $\tilde{P}_{Y^{ml}, Z^{ml}|W}$, and $(c)$ is because of the following bound:
\begin{align*}
     I(X^{ml} ; Y^{ml}| Z^{ml})&= H(Y^{ml}|Z^{ml})-H(Y^{ml}|Z^{ml},X^{ml})\\
    &= \sum_{i=1}^{m}H(\uY_i|Z^{ml},\uY^{i-1})-\sum_{i=1}^{m}H(\uY_i|X^{ml},Z^{ml},\uY^{i-1})\notag\\
    &\stackrel{(d)}{\leq} \sum_{i=1}^{m}H(\uY_i|\uZ_i)-\sum_{i=1}^{m}H(\uY_i|\uX_i,\uZ_i)\notag\\
    &= \sum_{i=1}^{m}I(\uX_i;\uY_i|\uZ_i) \leq m\sup_{\uX} I(\uX;\uY|\uZ)
     = ml(\epw-\epm),
\end{align*}
where  $(d)$ is due to the fact that $\uY_i$ is independent of $\uX^{i-1}, \uX_{i+1}^m, \uZ^{i-1},\uZ_{i+1}^m$ given $\uX_i$ and $\uZ_i$.
By normalizing and taking limits on both sides of \eqref{ineq:conv_14}, we get $R \leq \trk+(\epw-\epm)$, which implies that 
\begin{align}\label{ineq:conv_ew leq em_3}
     C_s(\bW, \trk) \leq  \trk+(\epw-\epm).
\end{align}
By combining \eqref{ineq:conv_ew leq em_1} and \eqref{ineq:conv_ew leq em_3}, we get the converse result $C_s(\bW, \trk) \leq \min\left\lbrace \trk+(\epw-\epm), 1-\epm \right\rbrace.$

\end{proof}

\bibliographystyle{IEEEtran}
\bibliography{IEEEabrv,References}

\end{document}